\documentclass[a4paper,UKenglish,cleveref, autoref, thm-restate]{lipics-v2021}

\pdfoutput=1 
\hideLIPIcs  

\title{Few Rows Tell Them Apart: Equivalence of Queries Mixing Set and Bag Semantics}

\titlerunning{Equivalence of Queries Mixing Set and Bag Semantics}

\author{Sara Cohen}{School of Computer Science and Engineering, The Hebrew University of Jerusalem, Israel}{scohen@mail.huji.ac.il}{https://orcid.org/0000-0002-8482-9435}{}

\authorrunning{S. Cohen}

\Copyright{Sara Cohen}



\category{} 

\relatedversion{} 

\funding{The authors were partially supported by the ISF (Israel Science Foundation, Grant 359/21).}

\nolinenumbers 
\usepackage{booktabs}

\newcommand{\adom}{\mathrm{adom}}
\newcommand{\cdb}[2]{\mathcal{D}_{#1}[#2]}
\newcommand{\pcdb}[2]{\mathcal{D}^{\pi}_{#1}[#2]}
\newcommand{\nmv}[1]{\#_M(#1)}
\newcommand{\myparagraph}[1]{\medskip\noindent\textbf{#1.}}
\newcommand{\NP}{\ensuremath{\mathsf{NP}}}

\definecolor{sqlkw}{rgb}{0.12,0.31,0.61}
\definecolor{sqlcomment}{rgb}{0.42,0.45,0.49}
\definecolor{sqlstr}{rgb}{0.55,0.10,0.10}
\definecolor{sqlbg}{rgb}{0.96,0.96,0.97}
\lstdefinestyle{sql}{%
  language=SQL,
  basicstyle=\small\ttfamily,
  backgroundcolor=\color{sqlbg},
  keywordstyle=\color{sqlkw}\bfseries,
  commentstyle=\color{sqlcomment}\itshape,
  stringstyle=\color{sqlstr},
  morekeywords={DISTINCT},
  deletekeywords={ORDER,COUNT},
  showstringspaces=false,
  columns=fullflexible,
  keepspaces=true,
  xleftmargin=1.2em,
  aboveskip=4pt, belowskip=2pt,
}

\AtEndEnvironment{example}{\unskip\nobreak\hfill\textcolor{lipicsGray}{\ensuremath{\lozenge}}}

\begin{document}

\maketitle

\begin{abstract}
Bounded SQL equivalence checkers search for a counter-example database of
bounded size, and a search that comes back empty proves nothing. We supply missing theory: computable bounds $B$
such that agreement on all databases with at most $B$ tuples per relation implies equivalence.
We work in the combined-semantics framework, which captures SQL's mix of duplicate-eliminating
(\texttt{DISTINCT}) and duplicate-preserving computation over set-valued relations. For
conjunctive queries we prove a bound linear in the query size for fixed multiset width:
inequivalent queries already disagree on a database with at most $2^{w}|Q|$ tuples, where the
width $w$ counts only the columns the queries actually read, independently of the total number
of multiset variables. Declared keys shrink the
bound to $2^{\mathit{kw}}|Q|$ for the smaller key-width $\mathit{kw}$, acyclic foreign keys
leave it unchanged, and the result extends to several classes of queries with comparisons, for which
equivalence had not previously been characterized. For these fragments, bounded search becomes
a terminating, complete decision procedure.
\end{abstract}

\section{Introduction}
\label{sec:intro}

Query equivalence is a central problem of database theory. Two queries are equivalent if they
return the same result on every database. Five decades of research on the problem have powered
query optimization, rewriting over materialized views, and data
integration~\cite{ChandraM77,ChaudhuriVardi93}; it has recently gained new urgency. Large
language models now rewrite SQL at scale, and err at scale: in a recent industrial study,
roughly a third of over $3{,}100$ LLM-proposed rewrites of production queries returned results
different from the original~\cite{QOVerify}, and text-to-SQL benchmarks' equivalence verdicts
determine published model rankings~\cite{SpotIt}.

This renewed demand has produced a wave of practical SQL equivalence checkers, split into two
one-sided families. Symbolic \emph{provers}~\cite{EQUITAS,SQLSolver,QED} certify equivalence,
but each only on a restricted syntactic fragment, and none certifies a refutation. Bounded
\emph{refuters}~\cite{Cosette,VeriEQL,Polygon,SpotIt} embrace the most expressive queries by
testing \emph{bounded equivalence}: a solver searches for a \emph{counter-example}---a database
on which the two queries disagree---among all databases with at most $n$ tuples per relation,
growing $n$ until timeout or fixing it in advance. A counter-example refutes equivalence; an
exhausted search proves nothing, yet its verdicts are trusted. A fifth of the pairs one suite of
VeriEQL's evaluation reports as ``checked'' are in fact inequivalent, two needing
counter-examples of over a thousand tuples~\cite{VeriEQL}, and SpotIt publishes what survives
its $n{=}5$-capped search as benchmark accuracy~\cite{SpotIt}.\looseness=-1

What is missing is a completeness threshold: a computable bound $B$ such that any two
inequivalent queries already disagree on some database with at most $B$ tuples per relation.
With such a bound the search stops at $B$ and reports \emph{equivalent} with certainty---and a
\emph{per-relation} version caps each relation separately, shrinking the search space further.
This paper establishes such bounds for the conjunctive core of the problem.\looseness=-1

We work in the framework of \emph{combined semantics}~\cite{Cohen06,Cohen09}, which captures
a core aspect of how real SQL queries evaluate. SQL mixes two modes of computation. A \texttt{DISTINCT} block or an
\texttt{EXISTS} subquery is a set computation, blind to duplicates, while projection without
\texttt{DISTINCT} preserves multiplicities. Combined semantics expresses this mix in one
language by declaring, per variable, whether different values contribute new copies to the
answer (\emph{multiset variables}) or not (\emph{set variables}). Following common practice, we
take stored relations to be sets, as is the case whenever every table declares a key.

For the two pure extremes, small counter-examples have long been known. Inequivalent set queries
(pure \texttt{DISTINCT}) disagree on a database with at most one tuple per query
atom~\cite{ChandraM77}, and the same holds for pure bag-set queries~\cite{ChaudhuriVardi93}. In
both cases the bound is the \emph{self-join size}: the maximum number of atoms sharing a
predicate. For queries that genuinely mix the two modes, the property fails, and it fails
already for the simplest possible queries.

\begin{example}
\label{ex:intro-no-small-ce}
The following queries return the names of \texttt{vip} customers. The first returns each name
once. The second returns each name once per distinct customer bearing it.
\begin{lstlisting}[style=sql]
QA: SELECT DISTINCT name FROM Customer WHERE type='vip'
QB: SELECT name FROM (SELECT DISTINCT cid, name FROM Customer WHERE type='vip') D
\end{lstlisting}
Each query uses the \texttt{Customer} table once, so the classical bounds promise a one-tuple
counter-example. None exists: the queries agree on every one-tuple database, and differ only
once two \texttt{vip} customers share a name---which takes two tuples. The classical bound is
wrong for mixed queries, and how far it must grow is precisely the question.
\end{example}

The example is benign, but the general phenomenon is not: it is not even clear that a computable
bound on counter-example size exists at all. Our main result is that it does, and it is linear
in the query size.

\myparagraph{Contributions}
\Cref{sec:bound} proves the central result, a counter-example bound for relational
combined-semantics queries that is \emph{linear in the query}: inequivalent queries disagree on
a database with at most $2^{w}|Q|$ tuples, where $|Q|$ is the number of atoms and the
\emph{multiset width} $w$ counts only the multiset columns the queries actually read.
\Cref{sec:keyed} admits integrity constraints. The bound refines per relation, drops to $2^{\mathit{kw}}|Q|$ for
the smaller \emph{key-width} $\mathit{kw}$ under declared keys, covering in particular queries
whose joins follow keys acyclically, and is unchanged under acyclic foreign keys. The same
argument characterizes equivalence there, as multiset-homomorphism of the key-chased reducts,
and so places the problem in \NP. These are the first results on equivalence of multiset and
combined queries in the presence of integrity constraints.
\Cref{sec:comparisons} extends the bound to several classes of queries with comparisons against
constants, for which no equivalence characterization, and hence no decision procedure, was
previously known.\looseness=-1

Taken together, the results turn heuristic bounded search into a terminating, complete
decision procedure for the conjunctive fragment studied here, with a threshold that is small whenever
self-joins are few and the multiset columns are narrow. All proofs, and some constructions, are deferred to the appendix.

\section{Preliminaries}
\label{sec:prelims}

We recall the framework of \emph{combined semantics} introduced in~\cite{Cohen06,Cohen09},
specialized to databases whose relations are sets. Combined semantics uniformly captures
both \emph{set semantics} (the \texttt{DISTINCT} behaviour of SQL, and the existential
behaviour of subqueries) and \emph{bag-set semantics} (the multiplicities produced by
projection without \texttt{DISTINCT}) within a single query language.\looseness=-1

\myparagraph{Databases}
Predicate symbols are denoted $p$, $q$, $r$. A \emph{database} $\mathcal D$ over a set of
predicate symbols $\mathcal P$ is a set of ground atoms with predicates from $\mathcal P$. Every
relation is thus a \emph{set}: multiplicities in query answers arise from the bag-set
behaviour of projection, not from duplicate stored tuples. The \emph{active domain}
$\adom(\mathcal D)$ is the set of constants occurring in $\mathcal D$.

\myparagraph{Keys and Foreign Keys}
Real schemas constrain their data. A \emph{schema} assigns every predicate $p$ of arity $r_p$ a
\emph{key} $\mathit{key}(p)\subseteq\{1,\dots,r_p\}$, and declares a set of \emph{foreign keys}.
A foreign key is a triple $(p,\bar\imath,q)$, where $\bar\imath$ is a tuple of positions of $p$
with $|\bar\imath|=|\mathit{key}(q)|$. A database $\mathcal D$ is \emph{legal} if it satisfies
every declared constraint. A key requires that no two distinct $p$-atoms agree on all positions in $\mathit{key}(p)$. A foreign key $(p,\bar\imath,q)$
requies for every $p(\bar s)\in\mathcal D$ some $q(\bar t)\in\mathcal D$ with
$\bar s|_{\bar\imath}=\bar t|_{\mathit{key}(q)}$. The foreign keys are \emph{acyclic} if the
relation ``$p$ references $q$'' on predicates has no cycle. When every key is the whole tuple,
i.e., $\mathit{key}(p)=\{1,\dots,r_p\}$ for all $p$, and no foreign keys are declared, every
database is legal. We call such a schema \emph{unconstrained}.

\begin{example}[Running Example]
\label{ex:db}
  Our examples use a small retail schema with four relations, key attributes underlined:
  \textsf{Customer}(\textit{\underline{cid}, name, email, phone, city, state, cc, signup, type}),
  \textsf{Order}(\textit{\underline{oid}, cid, date, via, status, shipcity, total}),
  \textsf{Item}(\textit{\underline{oid}, \underline{sku}, qty, price, discount}), and
  \textsf{Product}(\textit{\underline{sku}, title, category, brand, list, color}).
  An order is placed by a customer, and its items reference products by their \textit{sku}, so
  the schema declares acyclic foreign keys, from $\mathsf{Order}.\mathit{cid}$ to
  \textsf{Customer}, from $\mathsf{Item}.\mathit{oid}$ to \textsf{Order}, and from
  $\mathsf{Item}.\mathit{sku}$ to \textsf{Product}. A customer or a product may recur across many
  rows---one shopper places many orders, and a popular product appears in many of them.
\end{example}

\myparagraph{Syntax of Queries}
We denote constants by $c, d$ and variables (ranging over constants) by $x$, $y$, $z$.
A \emph{term} $s, t$ is either a variable or a constant. We use $\bar x$ and $\bar s$ to
denote a sequence of variables and terms, respectively. Queries can contain two types of
atoms:
\begin{itemize}
    \item \emph{relational atoms}, denoted $p(\bar s)$,
    \item \emph{order atoms} or \emph{comparisons}, denoted $s_1\,\rho\, s_2$ where
          $\rho\in\{<,\leq,>,\geq,=,\neq\}$.
\end{itemize}
A \emph{condition} $L$ is a conjunction of relational and order atoms.

Variables in a query are either \emph{distinguished}, appearing in the head, or
\emph{nondistinguished}, appearing only in the body. The nondistinguished variables are of two
types, \emph{set variables} and \emph{multiset variables}: intuitively, assignments differing on
multiset variables contribute to the multiplicity of the query result, while those differing on
set variables do not. We specify the set of multiset variables to the right of the condition, and
we call a column holding a multiset variable a \emph{multiset column}.

A \emph{conjunctive query} is an expression of the form $Q(\bar{x}) \leftarrow L, M$
    where $L$ is a conjunction of atoms and $M$ is the set of multiset variables. We require $M$
  to contain only nondistinguished variables. We assume queries are
  safe~\cite{ullman1988principles2}, i.e., every variable in the query occurs in some
  relational atom. For a query $Q$,
  we write $|Q|$ for the number of relational atoms of $Q$ and, when convenient, treat the body
  as this set, writing $p(\bar s)\in Q$ for a relational atom of $Q$.

We define several classes of queries. Query $Q$ is a \emph{set query} if
$M = \emptyset$, a \emph{multiset query} if $M$ is precisely the set of all
nondistinguished variables, and a \emph{combined query} otherwise. A query is \emph{relational} if it has no order atoms. Set queries
correspond to SQL queries that eliminate duplicates (via \texttt{DISTINCT}), while multiset queries correspond to SQL queries without \texttt{DISTINCT} and with no existential subqueries. 
Set and multiset variables also meet through existential subqueries: a multiset query that probes
a multiset column with an \texttt{EXISTS}, \texttt{IN}, or \texttt{ANY} subquery equates the
multiset variable with a set variable of the inner block. The value is then counted once, however
many rows of the inner block witness the subquery. Real SQL queries may combine both set and
multiset computations, as demonstrated in the following example.
\begin{example}[Query syntax]
  \label{ex:syntax}
  Recall the relations of \cref{ex:db}. The queries below report the customers who
  ordered an electronics product, differing only in how they treat duplicates:
\begin{lstlisting}[style=sql]
Q1: SELECT D.cid
    FROM   (SELECT DISTINCT O.cid, I.sku
            FROM   Order O, Item I, Product P
            WHERE  O.oid = I.oid AND I.sku = P.sku
                   AND P.category = 'electronics') D
Q2: SELECT DISTINCT O.cid
    FROM   Order O, Item I, Product P
    WHERE  O.oid = I.oid AND I.sku = P.sku AND P.category = 'electronics'
\end{lstlisting}
  Query \texttt{Q1} deduplicates on $(\mathit{cid},\mathit{sku})$ before projecting onto
  $\mathit{cid}$. Query \texttt{Q2} deduplicates fully. Query \texttt{Q3}, defined as \texttt{Q2}
  without \texttt{DISTINCT}, returns one row per electronics line item. To render these in our notation,
  we write $\_$ for an \emph{anonymous} variable (one used nowhere else, distinct at each appearance) and
  $\dots$ for a run of them, and let $L$ be the body
  \[
    \mathtt{Order}(o,x,\dots)\,\wedge\,\mathtt{Item}(o,k,\dots)\,\wedge\,
    \mathtt{Product}(k,\_,\mathtt{electronics},\dots),
  \]
  finding a customer $x$ who placed an order $o$ containing an electronics product $k$. Writing $V$ for
  the set of \emph{all} nondistinguished variables of $L$ ($o$, $k$, and every $\_$), the
  three queries are $Q_1(x)\leftarrow L,\,\{k\}$, \ $Q_2(x)\leftarrow L,\,\emptyset$, and
  $Q_3(x)\leftarrow L,\,V$.
  Thus $Q_1$ is combined, counting a customer once per distinct product, $Q_2$ is a set query,
  and $Q_3$ is a multiset query.
\end{example}

\myparagraph{Semantics of Queries}
Let $\mathcal D$ be a database and $Q(\bar x)\leftarrow L,M$ be a query. A \emph{satisfying
assignment} $\gamma$ maps the terms of $Q$ to constants such that
\begin{itemize}
    \item $\gamma$ is the identity on constants,
    \item for every relational atom $p(\bar s)$ in $L$, we have $p(\gamma \bar s)\in \mathcal D$,
    \item for every order atom $s_1\,\rho\, s_2$ in $L$, we have $\gamma s_1\,\rho\, \gamma s_2$.
\end{itemize}

Let $\Gamma(Q,\mathcal{D})$ denote the set of satisfying assignments of $Q$. If $Y$ is a set
of variables of $Q$, we write $\Gamma_Y(Q,\mathcal{D})$ for the projection of
$\Gamma(Q,\mathcal{D})$ onto $Y$, i.e., the set of all assignments to $Y$ that extend to a
satisfying assignment in $\Gamma(Q,\mathcal{D})$. Let $X$ be the set of distinguished variables
of $Q$. The \emph{result} of applying $Q$ to $\mathcal{D}$, denoted $Q(\mathcal D)$, is the
multiset
  \[ \bigl\{\!\bigl\{\, \gamma(\bar{x}) \mid \gamma \in \Gamma_{X\cup M}(Q,\mathcal{D}) \,\bigr\}\!\bigr\}.
  \]
  Thus each assignment to the distinguished and multiset variables that extends to a satisfying
  assignment contributes one tuple to the answer, whereas the choices for the set variables are
  collapsed. When $Q$ is a set query this coincides with set semantics, as $M=\emptyset$ and
  every answer tuple appears exactly once. When $Q$ is a multiset query it coincides
  with bag-set semantics.\looseness=-1

The answer multiplicities are exactly what SQL's \texttt{COUNT} reports, \texttt{COUNT(*)} for a
bag block and \texttt{COUNT(DISTINCT ...)} over the multiset columns, and adding such an
aggregate preserves equivalence. Combined-semantics equivalence is thus precisely the
equivalence of SQL counting queries.

\myparagraph{Query Containment and Equivalence}
\phantomsection\label{sec:equiv}
Fix a schema. Query $Q$ is \emph{contained} in query $Q'$, denoted $Q\subseteq Q'$, if
$Q(\mathcal D)$ is a subbag of $Q'(\mathcal D)$ for every \emph{legal} database $\mathcal D$. Queries $Q$ and $Q'$ are \emph{equivalent}, written $Q\equiv Q'$, if containment holds in both
directions. Over an unconstrained schema these are the classical notions, quantifying over all
set databases.

Over unconstrained schemas, containment of set queries and equivalence of multiset queries are
classically characterized by homomorphisms and isomorphisms respectively~\cite{ChandraM77,Klug,vanderMeyden,ChaudhuriVardi93,EquivAggregate} (see
\cref{sec:related}). For queries combining set and multiset variables, equivalence is
characterized only for relational queries over unconstrained schemas, by \emph{multiset
homomorphisms}~\cite{Cohen06,Cohen09}. Almost no characterization is known with comparisons (except
for the join queries of~\cite{Cohen06}), and no
prior work addresses combined equivalence under integrity constraints.

Unlike work that directly characterizes equivalence, we reduce equivalence to \emph{bounded}
equivalence. Formally, let $\mathcal D$ be a database over predicates $\mathcal P$.
The \emph{size} of a predicate $p$ in $\mathcal D$, denoted $\mathit{size}(p,\mathcal D)$, is the
number of $p$-facts in $\mathcal D$. We say that $\mathcal D$ is \emph{$n$-bound} if
$\mathit{size}(p,\mathcal D)\leq n$ for all $p\in \mathcal P$. We write $Q\subseteq_n Q'$ if
$Q(\mathcal D)\subseteq Q'(\mathcal D)$ for all legal $n$-bound databases $\mathcal D$, and
$Q\equiv_n Q'$ if this holds in both directions. Clearly, if $Q\subseteq Q'$ then
$Q\subseteq_n Q'$ for every $n$. This paper studies the converse:
  \emph{Given a class of queries, is there a computable bound $n$ such that $n$-bounded
  equivalence implies equivalence?}

Under both set and bag-set semantics, over unconstrained schemas, the equivalence problem enjoys
the \emph{small counter-example property}: if two queries are inequivalent, a witnessing database
exists whose size is bounded by the number of atoms in the
queries~\cite{ChandraM77,ChaudhuriVardi93,EquivAggregate,Cohen09}.
To be precise, the \emph{self-join size} $\mathit{sj}(Q)$ of $Q$ is the maximum number of
occurrences of any single predicate in $Q$. Write $k=\mathit{sj}(Q)$ and $k'=\mathit{sj}(Q')$.
For set queries, $Q\subseteq Q'$ if and only if $Q\subseteq_k Q'$, since the frozen body of
$Q$---a database with one atom per subgoal---already witnesses any failure of a containment
mapping~\cite{ChandraM77}. For multiset queries, $Q\equiv Q'$ if and only if
$Q\equiv_{\max(k,k')} Q'$, by the characterization of bag-set equivalence as isomorphism of the
queries~\cite{ChaudhuriVardi93}, which a database built from the bodies likewise detects. For
combined queries, bounded equivalence up to the self-join size no longer implies equivalence,
even for very simple queries, as demonstrated by \cref{ex:intro-no-small-ce}.

This paper extends these bounds to queries that mix set and multiset variables. Since such
queries cannot be equivalent unless their bodies are equivalent as set queries, we assume
throughout that the queries are \emph{set-equivalent} over the legal databases of the schema at
hand, and that each is \emph{satisfiable}, nonempty on some legal database. We focus on whether
their multiplicities coincide. Set-equivalence is itself decided within thresholds below every bound we prove~\cite{ChandraM77,AhoSagivUllman79}.

\section{Unconstrained Schemas}
\label{sec:bound}
\phantomsection\label{sec:relational}

This section proves the counter-example bound for relational queries over an unconstrained
schema (\cref{thm:main-bound}), where every database is legal. Constrained schemas follow in
\cref{sec:keyed}. Recall the standing assumption that $Q,Q'$ are set-equivalent
(\cref{sec:equiv}), so their answers can differ only in multiplicity.

Throughout this section queries are relational (no order atoms), the setting of the~\cite{Cohen06,Cohen09} characterization.
Let $Q(\bar x)\leftarrow L,M$ and $Q'(\bar x')\leftarrow L',M'$ be relational conjunctive queries.
A \emph{multiset-homomorphism} from $Q'$ to $Q$ is a mapping $\mu$ from the terms of $Q'$ to the
terms of $Q$ satisfying:
\begin{enumerate}
\item $\mu(\bar x')=\bar x$,
\item $\mu$ maps each constant to itself,
\item every atom of $\mu(L')$ occurs in $L$, and
\item $\mu$ maps $M'$ injectively into $M$: $\mu(M')\subseteq M$ and $\mu y\neq \mu y'$ for
      distinct $y,y'\in M'$.
\end{enumerate}
Conditions (1)--(3) state that $\mu$ is an ordinary containment mapping from $Q'$ to $Q$, which
witnesses the set-containment $Q\subseteq Q'$. We say that $Q$ and $Q'$ are \emph{multiset-homomorphic} if there is a
multiset-homomorphism from $Q'$ to $Q$ \emph{and} one from $Q$ to $Q'$. Since each direction
injects one set of multiset variables into the other, multiset-homomorphic queries necessarily
satisfy $|M|=|M'|$.

Multiset-homomorphisms characterize equivalence of relational queries.
\begin{theorem}[\cite{Cohen06,Cohen09}]
\label{thm:relational-equiv}
Let $Q$ and $Q'$ be relational conjunctive queries. Then $Q\equiv Q'$ if and only if $Q$ and $Q'$
are multiset-homomorphic.
\end{theorem}

Sufficiency is immediate: a multiset-homomorphism from $Q'$ to $Q$ shows that $Q'$ returns each
tuple at least as often as $Q$ on every database, so maps in both directions give equality.
Necessity is proved by a database construction that we now outline, as the same construction
yields our counter-example bound.

\myparagraph{The canonical family and the characteristic monomial}
\phantomsection\label{sec:cmono}
Let $m=|M|$ be the number of multiset variables of $Q$. The \emph{canonical family} of $Q$ is
parameterized by a vector $\vec N=(N_1,\dots,N_m)\in\mathbb N_+^m$, one parameter per multiset
variable. In the database $\cdb{\vec N}{Q}$, the $j$-th multiset variable ranges over a fresh
domain of exactly $N_j$ constants, while the distinguished and set variables receive fixed fresh
values, shared across all assignments. For each of the resulting $\prod_{j=1}^m N_j$ assignments,
one ground atom is inserted per subgoal of $Q$, and identical atoms are merged. Using $\bar c$
for the fixed image of the head, we have $\bar c\in Q(\cdb{\vec N}{Q})$ for every $\vec N$.\looseness=-1

For any query $R$, let $\mathcal F^{(R)}(\vec N)$ be the multiplicity of $\bar c$ in the bag
$R(\cdb{\vec N}{Q})$. The function $\mathcal F^{(R)}$ is a polynomial in $N_1,\dots,N_m$ with
non-negative integer coefficients, of total degree at most the number of multiset variables of
$R$, since each multiset variable of $R$ contributes at most one parameter factor. One monomial
of $\mathcal F^{(Q)}$ is special, the \emph{characteristic monomial}
$\mathcal P^{(Q)}_*=N_1 N_2\cdots N_m$: the unique monomial of total degree $m$ that is
\emph{squarefree}, using every parameter exactly once. It always appears in $\mathcal F^{(Q)}$
with a positive coefficient, realized by the generic assignment that places each multiset
variable in its own fresh domain. Crucially, $\mathcal P^{(Q)}_*$ appears in
$\mathcal F^{(Q')}$ if and only if there is a multiset-homomorphism from $Q'$ to
$Q$~\cite{Cohen09}.\looseness=-1

\Cref{thm:relational-equiv} follows. If $Q\equiv Q'$ then $\mathcal F^{(Q)}=\mathcal F^{(Q')}$,
so the characteristic monomial---present in $\mathcal F^{(Q)}$---also appears in
$\mathcal F^{(Q')}$, yielding a multiset-homomorphism from $Q'$ to $Q$, and symmetrically over
the family of $Q'$. The same correspondence separates inequivalent queries on the family
itself, and the rest of this section turns that separation into an explicit small witness.

\myparagraph{Separation at a Boolean corner}
\phantomsection\label{sec:sep}
Throughout, for a query $R$ with multiset variables $M_R$, we write $m_R:=|M_R|$, so $m=m_Q$. The
engine of the bound is the next proposition: it separates the two multiplicity functions
\emph{on the canonical family itself}, not merely on some unspecified database. Its proof rests on
three facts restated from~\cite{Cohen09}: polynomiality with a per-query
degree bound, positivity of the characteristic monomial in $\mathcal{F}^{(Q)}$, and the
multiset-homomorphism criterion for its appearance in $\mathcal{F}^{(Q')}$.

\begin{restatable}[Characteristic monomial separation]{proposition}{cmonosep}
\label{prop:cmono-sep}
Let $Q'$ be a relational query with no multiset-homomorphism from $Q'$ to $Q$. Then the characteristic monomial
$\mathcal{P}^{(Q)}_*$ does not appear in $\mathcal{F}^{(Q')}$, so
$P := \mathcal{F}^{(Q)} - \mathcal{F}^{(Q')}$ is a nonzero polynomial on the family
$\{\cdb{\vec N}{Q}\}$ (its characteristic monomial coefficient is positive). If moreover $m_{Q'} \le m$, then
$P$ has total degree at most $m$.
\end{restatable}

This says more than that $P$ is a nonzero polynomial of degree at most $m$. The monomial
responsible is the multilinear $N_1\cdots N_m$, and by Alon's Combinatorial
Nullstellensatz~\cite{Alon99}, a nonzero top-degree coefficient on $\prod_i N_i^{t_i}$ forces a
nonzero point on any grid $A_1\times\cdots\times A_m$ with $|A_i|>t_i$. Here every $t_i$ is $1$,
so we may take every $A_i=\{1,2\}$. Hence, when no multiset-homomorphism $Q'\to Q$ exists and
$m_{Q'}\le m$, the queries already disagree on a \emph{corner database} $\cdb{\vec a}{Q}$ with
$\vec a\in\{1,2\}^m$. The second requirement is free, as the queries can always be named to meet
it: when their counts differ, take the one with fewer multiset variables as $Q'$, and no monomial
of $\mathcal F^{(Q')}$ then reaches the degree $m$ of $\mathcal P^{(Q)}_*$, so that direction
fails outright. Note that a generic degree-$m$ polynomial would force only the grid
$\{1,\dots,m+1\}^m$ and a witness of size $|Q|(m+1)^m$. The multilinearity of the characteristic
monomial pins the witness to the corner. It remains to bound the size of the corner databases.\looseness=-1

\myparagraph{The size of a corner database}
We measure database size \emph{per relation}, as in the $n$-bound notion of \cref{sec:prelims}:
$\mathit{size}(p,\mathcal D)$ is the number of $p$-facts. Recall that $|Q|$ is the number of
relational atoms of $Q$. The construction inserts, for each of the
$\prod_i N_i$ blow-up assignments, at most one atom per subgoal, so, before merging, a relation
could carry a factor $\prod_i N_i$. At a Boolean corner $\vec a\in\{1,2\}^m$, however, the
merging of identical atoms is dramatic. Each subgoal produces only a bounded number of distinct
atoms, no matter how many multiset variables are set to two. The relevant parameter is not the
arity but the number of multiset variables an atom carries.\looseness=-1

For an atom $p(\bar s)$, let $\nmv{p(\bar s)}$ be the number of distinct multiset variables
occurring in $\bar s$. Distinguished and set variables, being fixed in the canonical family, do not
count. For a predicate $p$, the \emph{width of $p$ in $Q$} is
$w(Q,p):=\max\{\nmv{p(\bar s)}:p(\bar s)\in Q\}$, the largest such count among atoms with
predicate $p$. The \emph{multiset width} of $Q$ is the maximum over predicates, $w(Q):=\max_p w(Q,p)$.
Thus $w(Q,p)\le w(Q)$ for every $p$, and $w(Q)$ is at most the maximum arity, often much smaller.
Recall the \emph{self-join size} $\mathit{sj}(Q)$, the maximum number of occurrences of any single
predicate in $Q$. We refine it per predicate: $\mathit{sj}(Q,p)$ is the number of subgoals of $Q$ with
predicate $p$, so $\mathit{sj}(Q)=\max_p\mathit{sj}(Q,p)$. When $Q$ is clear from context we drop it,
writing $w$, $w(p)$, $\mathit{sj}$, and $\mathit{sj}(p)$.

\begin{restatable}[Distinct atoms at a corner]{lemma}{cornersize}
\label{lem:corner-size}
For every $\vec a\in\{1,2\}^m$ and predicate $p$,
\[
\mathit{size}\bigl(p,\cdb{\vec a}{Q}\bigr)\;\le\;2^{w(p)}\,\mathit{sj}(p),
\qquad\text{hence}\qquad
|\cdb{\vec a}{Q}|\;\le\;\sum_{p(\bar s)\in Q}2^{\,\nmv{p(\bar s)}}\;\le\;2^{w}\,|Q| .
\]
\end{restatable}

The reason is that a subgoal's atoms vary only through its multiset variables, each of which
takes at most two values at a corner. Neither bound mentions the corner $\vec a$ or the number
$m$ of multiset variables, so both are uniform over the corner and independent of $m$.
Combining the separation with the size count gives the main theorem.

\begin{restatable}[Counter-example bound for inequivalence]{theorem}{mainbound}
\label{thm:main-bound}
Let $Q$ and $Q'$ be relational queries with $Q \not\equiv Q'$. Write
$w_p:=\max(w(Q,p),w(Q',p))$ and $w:=\max_p w_p=\max(w(Q),w(Q'))$ for the per-predicate and the
overall width of the pair. Then there is a database $\mathcal D$ with
$Q(\mathcal D) \neq Q'(\mathcal D)$,
\[
|\mathcal D| \;\le\; 2^{w}\max\bigl(|Q|,|Q'|\bigr)
\qquad\text{and}\qquad
\mathit{size}(p,\mathcal D)\;\le\;2^{w_p}\max\bigl(\mathit{sj}(Q,p),\mathit{sj}(Q',p)\bigr)\ \text{for every }p.
\]
In particular, $Q\equiv_{b}Q'$ implies $Q\equiv Q'$ for
$b=\max_p 2^{w_p}\max(\mathit{sj}(Q,p),\mathit{sj}(Q',p))$.
\end{restatable}

By \cref{thm:relational-equiv}, some direction, say $Q'\to Q$, admits no
multiset-homomorphism, named as above. The separation then produces a corner database of $Q$ on which the queries disagree, and
\cref{lem:corner-size} bounds its size.

Two tightness questions must be kept apart: whether every multiset variable must be doubled, and
whether the $2^{w}|Q|$ atoms the bound then permits are ever needed. The first is settled by one
pair of queries, which leaves the second untouched. The exact threshold is open
(\cref{sec:conclusion}).

\begin{example}[The full corner can be necessary]
\label{ex:corner-needed}
Let $L$ be the directed $4$-cycle $p(x_0,x_1)\wedge p(x_1,x_2)\wedge p(x_2,x_3)\wedge p(x_3,x_0)$
and take the Boolean queries $Q\leftarrow L,\{x_0,x_1\}$ and $Q'\leftarrow L,\{x_0,x_2\}$,
differing only in their multiset variables. On the canonical family
$\mathcal F^{(Q)}-\mathcal F^{(Q')}=(N_1+1)(N_2+1)-2(N_1+N_2)=(N_1-1)(N_2-1)$, which vanishes
whenever $N_1=1$ or $N_2=1$. Only the all-twos corner separates, so no multiset variable escapes
doubling. The \emph{size} bound is loose on the same pair. \Cref{thm:main-bound} permits $16$
atoms, the witness $\cdb{(2,2)}{Q}$ uses nine, and $\{p(0,0),p(0,1),p(1,0)\}$ already separates.
\end{example}

\myparagraph{Collapsing unobserved attributes}
The width $w$ can come close to the arity when a wide table's columns each carry their own
multiset variable. Yet a query typically inspects only a handful of
those columns, and the width that governs the bound ought to count only them. Call the $i$-th attribute
of a relation $r$ \emph{unobserved} (for the pair $Q,Q'$) if in neither query does any atom of $r$
place, in position $i$, a distinguished variable, a \emph{join variable}---one occurring in two or more
positions of the body---or a variable occurring in a comparison. A multiset variable in an unobserved
attribute feeds the multiplicity but is otherwise inert. The \emph{collapse} of the pair $Q,Q'$
merges each relation's unobserved attributes into one fresh attribute, rewriting queries and
databases alike. It yields queries $Q^{\downarrow},Q'^{\downarrow}$ over a
shared schema with $w(Q^{\downarrow})\le w(Q)$, $|Q^{\downarrow}|=|Q|$, and
$\mathit{sj}(Q^{\downarrow})=\mathit{sj}(Q)$, and it preserves every relation's size.

\begin{restatable}[Collapse preserves equivalence]{lemma}{collapselem}
\label{lem:collapse}
$Q\equiv Q'$ if and only if $Q^{\downarrow}\equiv Q'^{\downarrow}$.
\end{restatable}

The merged attribute carries a fresh multiset variable exactly when a merged column did, and that
variable ranges over the product of the merged domains, so it reproduces exactly the multiplicity
factor the merged variables contributed. Because the collapse preserves sizes, witnesses for
$Q\not\equiv Q'$ and for $Q^{\downarrow}\not\equiv Q'^{\downarrow}$ correspond with identical
per-relation sizes. \Cref{thm:main-bound} holds for any relational queries, so we may apply it
to the collapse $Q^{\downarrow},Q'^{\downarrow}$ and lift the witness back by
\cref{lem:collapse}, replacing the width $w$ by the collapsed width
$\max\bigl(w(Q^{\downarrow}),w(Q'^{\downarrow})\bigr)\le w$.

\begin{example}[Collapsing a wide pair]
\label{ex:collapse-wide}
Collapse the pair $(Q_1,Q_3)$ of \cref{ex:syntax}. In their shared body the only joins are $o$
and $k$, with head $x$, so every other column holds an anonymous variable occurring once, hence
unobserved. The collapse merges each relation's unobserved columns into one fresh attribute,
giving the body
$\mathtt{Order}^{\downarrow}(o,x,z_1)\wedge
\mathtt{Item}^{\downarrow}(o,k,z_2)\wedge
\mathtt{Product}^{\downarrow}(k,\mathtt{electronics},z_3)$.
In $Q_3^{\downarrow}$ the new variables $z_1,z_2,z_3$ are multiset and the width drops from $6$
to $3$. In $Q_1^{\downarrow}$ they are set placeholders and the width stays $1$. The canonical
witness thus shrinks from $2^{6}|Q_3|$ to $2^{3}|Q_3|$ atoms: collapsing first saves
exponentially in the columns the queries ignore.
\end{example}

\section{Constrained Schemas}
\label{sec:keyed}

We now let the schema constrain its data, first by declared keys and then, at the end of the
section, by acyclic foreign keys as well (\cref{sec:fk}). Legality means satisfying the
declared constraints. A key need not remove any counter-example, as is the case with the two-customer witness of \cref{ex:intro-no-small-ce} which respects the key of \textsf{Customer}. It can, however,
remove all of them:
under $\mathit{key}(p)=\{1\}$, the $4$-cycle pair of \cref{ex:corner-needed} becomes equivalent,
as the key makes $p$ a partial function $f$ and both queries count the points with
$f^{4}(x_0)=x_0$.

For an atom $p(\bar s)$ of $Q$, let $\mathit{kw}(p(\bar s))$ be the number of distinct multiset
variables of $Q$ occurring at a \emph{key} position of the atom, i.e., at a position in
$\mathit{key}(p)$. The \emph{key-width of a predicate $p$ in $Q$} is
$\mathit{kw}(Q,p):=\max\{\mathit{kw}(p(\bar s)):p(\bar s)\in Q\}$, and the
\emph{key-width} of $Q$ is $\mathit{kw}(Q):=\max_p\mathit{kw}(Q,p)$. As with the multiset width,
we drop $Q$ when it is clear from context. Only key positions are counted, so
$\mathit{kw}(Q)\le w(Q)$, with equality when every key is the whole tuple. In the case
where multiset variables never occupy a key position, $\mathit{kw}(Q)=0$.\looseness=-1

Key-width governs the bound under keys because a multiset variable off the key
carries no multiplicity. Each atom $p(\bar s)$ of $Q$ reads its key as a functional dependency on the
variables of $Q$: the variables at the key positions of the atom determine those at its
remaining positions. For a set $T$ of variables, $T^{+}$ is the closure of $T$ under these
dependencies, and a variable $v$ is \emph{$T$-determined} if $v\in T^{+}$: on any legal database
its value is fixed once the variables of $T$ are.

\begin{restatable}[Demotion]{lemma}{demotionlem}
\label{lem:demotion}
Let $X$ be the distinguished variables of $Q$ and let $v\in M$ be $(X\cup(M\setminus\{v\}))$-determined.
Let $Q^{-v}$ be $Q$ with $v$ moved from $M$ to the set variables. Then $Q\equiv Q^{-v}$ and
$\mathit{kw}(Q^{-v})\le\mathit{kw}(Q)$.
\end{restatable}

On a legal database the forgotten variable is determined by the rest, so forgetting it bijects
the assignment sets and no multiplicity changes. Demote an
$(X\cup(M\setminus\{v\}))$-determined multiset variable, repeat until none remains, and call
the result a \emph{key-reduct} $Q^{\flat}$. By \cref{lem:demotion}, $Q^{\flat}\equiv Q$
and $\mathit{kw}(Q^{\flat})\le\mathit{kw}(Q)$. Its surviving multiset variables are \emph{key-free}:
none is determined by the remaining distinguished and multiset variables. Demotion also explains
why keys force the combined setting on us. A multiset query need not stay one over legal
databases, its key-reduct having genuine set variables, so bag-set semantics alone does not
survive the declaration of a key.

Demotion is non-deterministic: different orders demote different mutually-determined variables,
and the surviving sets can differ even in size and key-width. But
every key-reduct satisfies $Q^{\flat}\equiv Q$ and $\mathit{kw}(Q^{\flat})\le\mathit{kw}(Q)$,
which is all the bound requires, and whether the reduct is key-anchored, defined below, is
order-independent (\cref{prop:anchored}).

\myparagraph{Anchoring and legality}
Key-freeness alone does \emph{not} make the canonical family of \cref{sec:cmono} legal: the demotion
that frees one variable can fix another that serves as a key elsewhere. The obstruction is a
cyclic key dependency.

\begin{example}[Key-freeness is not enough]
\label{ex:cyclic-key}
Let $R$ have body $p(y_0,y_1)\wedge p(y_1,y_0)$, with $\mathit{key}(p)=\{1\}$ and
$M=\{y_0,y_1\}$. The key makes $p$ a partial function, so the two variables move as a pair.
Demotion frees $y_0$ by freezing $y_1$, and the canonical family then varies $y_0$ alone, which
the key forbids. Legal witnesses of the same multiplicity still exist, so only the rigid family
fails.
\end{example}

What the construction needs is
\emph{anchoring}. Call a query $R$ \emph{key-anchored} if every multiset variable of $R$ occurs, in
every atom of $R$ containing it, at a key position of that atom. The \emph{key-chase} of $R$
repeatedly unifies the terms of two atoms of a predicate that agree on their key, and merges
duplicates. It terminates and preserves equivalence over legal databases, and the canonical
family of a key-chased key-anchored query is legal at every corner (\cref{lem:legal-canonical}). Call $Q$ \emph{key-anchorable} if its key-reducts are
key-anchored. The condition is decidable and order-independent, being equivalent to
$M\subseteq(X\cup A)^{+}$, where $A$ collects the multiset variables occurring only at key
positions (\cref{prop:anchored}). Finally, in a legal
database distinct atoms of a predicate differ on their key, so a legal corner database has at
most $2^{\mathit{kw}(p)}$ atoms per subgoal (\cref{lem:keyed-corner-size}). The bound follows for
key-anchorable queries.

\begin{restatable}[Equivalence and counter-examples under declared keys]{theorem}{keyedbound}
\label{thm:keyed-bound}
Let $Q,Q'$ be key-anchorable relational queries over a schema with declared keys and no foreign
keys, and write $Q^{\natural}$ for the key-chase of the key-reduct of $Q$. Put
$\mathit{kw}:=\max(\mathit{kw}(Q^{\natural}),\mathit{kw}(Q'^{\natural}))$ and
$\mathit{kw}(p):=\max(\mathit{kw}(Q^{\natural},p),\mathit{kw}(Q'^{\natural},p))$, both at most
the same quantities for $Q$ and $Q'$. Then
\begin{enumerate}
\item $Q\equiv Q'$ over legal databases if and only if $Q^{\natural}$ and $Q'^{\natural}$ are
multiset-homomorphic, and
\item if $Q\not\equiv Q'$, there is a legal database $\mathcal D$ with
$Q(\mathcal D)\neq Q'(\mathcal D)$,
\[
|\mathcal D|\le 2^{\mathit{kw}}\max(|Q^{\natural}|,|Q'^{\natural}|)
\quad\text{and}\quad
\mathit{size}(p,\mathcal D)\le 2^{\mathit{kw}(p)}\max(\mathit{sj}(Q^{\natural},p),\mathit{sj}(Q'^{\natural},p))
\]
for every predicate $p$.
\end{enumerate}
In particular $Q\equiv_{b}Q'$ implies $Q\equiv Q'$ for
$b=\max_p 2^{\mathit{kw}(p)}\max(\mathit{sj}(Q^{\natural},p),\mathit{sj}(Q'^{\natural},p))$.
\end{restatable}

Both clauses come from one argument. The proof passes to the key-reducts, which are equivalent
to the original queries by \cref{lem:demotion}. Their corner databases are legal by
\cref{lem:legal-canonical}, so the separation of \cref{thm:main-bound} applies to them verbatim
and delivers a \emph{legal} witness, which \cref{lem:keyed-corner-size} counts. Multiset-homomorphic
reducts are equivalent over all databases by \cref{thm:relational-equiv}, hence over the legal
ones, which gives clause~(1) with the separation as its converse.

One consequence: if every multiset variable is $X$-determined the reducts are set queries, which
the standing assumption already equates, so the question bites only when some multiset variable
is not forced by the head. Finally, the collapse of \cref{lem:collapse} carries over
to keyed schemas, replacing $\mathit{kw}$ in \cref{thm:keyed-bound} by its collapsed counterpart.

The anchoring requirement is not incidental. Call an atom of $Q$ a \emph{counting atom} if all
its variables are distinguished or multiset, and call $Q$ \emph{well-formed} if every multiset
variable occurs in some counting atom. Real SQL queries whose counting arises from a
\texttt{FROM} clause over base tables are well-formed, every column of such a table being either
returned or a contributor to the multiplicity. Say that a counting atom $b$
\emph{references} a counting atom $b'$ if some multiset variable at a key position of $b'$
occurs at a non-key position of $b$.\looseness=-1

\begin{restatable}[Well-formed queries with acyclic key references]{proposition}{wfanchored}
\label{prop:wf-anchored}
Let $Q$ be a well-formed query, over a schema with declared keys, such that
\begin{enumerate}
\item every occurrence of a multiset variable at a non-key position lies in a counting atom, and
\item the references among the counting atoms of $Q$ are acyclic.
\end{enumerate}
Then $M\subseteq(X\cup A)^{+}$, that is, $Q$ is key-anchorable.
\end{restatable}

Two such queries meet the requirements of \cref{thm:keyed-bound}. In SQL terms, the two conditions
say that multiset columns are compared, inside subqueries, only against key columns, and that the
join conditions of the counting block do not cycle through keys. Joins that follow foreign keys,
as in star and snowflake schemas, always satisfy both. Neither condition is free: the query of
\cref{ex:cyclic-key} is well-formed, both its atoms being counting atoms, yet its two atoms
reference each other and it is not key-anchorable.

\myparagraph{Foreign keys}
\phantomsection\label{sec:fk}
We now admit the full schemas of \cref{sec:prelims}: declared keys together with acyclic
foreign keys. Foreign keys are inclusion dependencies,
handled by a chase. Because they point key-into-key, the \emph{foreign-key chase} is
deterministic: for each unmatched reference it adds the referenced atom, copying the key values
and filling the remaining positions with fresh \emph{set} variables (\cref{lem:fk-chase}). We
write $\widehat Q$ for the chased query. For acyclic foreign keys the chase terminates, adds no
multiset variable, so $M$, $w$, and $\mathit{kw}$ are unchanged, and its added subgoals are
redundant on legal databases, so $Q$ and $\widehat Q$ return the same answers. As with demotion,
the chase moves a multiset query out of its class: the atoms it adds carry set variables, so
$\widehat Q$ is combined even when $Q$ was not.

\begin{restatable}[Counter-example bound under keys and acyclic foreign keys]{theorem}{fkbound}
\label{thm:fk-bound}
Let $Q,Q'$ be relational queries over a schema with keys and acyclic foreign keys, such that
$\widehat Q,\widehat{Q'}$ are key-anchorable. Then both clauses of \cref{thm:keyed-bound} hold
for $\widehat Q,\widehat{Q'}$: equivalence over legal databases is multiset-homomorphism of
$\widehat Q^{\natural}$ and $\widehat{Q'}^{\natural}$, and inequivalent $Q,Q'$ are separated by a
legal database with at most $2^{\mathit{kw}}\max(|\widehat Q|,|\widehat{Q'}|)$ atoms. The query
size $|\widehat Q|$ is at most $|Q|$ times the reference depth of the schema.
\end{restatable}

The chase makes every referenced atom explicit, so the witness of \cref{thm:keyed-bound} for the
chased queries is already closed under the foreign keys. The chase also preserves the
requirements of \cref{prop:wf-anchored}, since every atom it adds places existing variables only
at key positions.\looseness=-1

Clause~(1) also settles the complexity, which the bounds leave open: their witnesses are
exponential, whereas a multiset-homomorphism is a mapping between queries of polynomial size.

\begin{restatable}[Complexity of equivalence under constraints]{corollary}{npcor}
\label{cor:np}
Deciding $Q\equiv Q'$ over a schema with declared keys and acyclic foreign keys is \NP-complete
for queries whose foreign-key chases are key-anchorable.
\end{restatable}

Membership guesses the two multiset-homomorphisms of clause~(1), checked by inspection of the
atoms, over queries computed in polynomial time: the foreign-key chase is bounded by the
reference depth, the key-reduct is a closure under the key dependencies, and the key-chase is a
union-find on the terms. Hardness is classical, already for $M=M'=\emptyset$~\cite{ChandraM77}.
Enumeration is thus needed only to \emph{exhibit} a counter-example, never to decide.

\myparagraph{The residual open case}
Outside key-anchorability the separation is driven by a multiset variable that demotion has
stranded off a key, so the rigid canonical family is illegal although legal witnesses exist
(\cref{ex:cyclic-key}). The regime is not exotic: counting the customers that placed some order
joins the counted $\mathit{cid}$ to a non-key position of an existential atom. What is missing is
\emph{realizability}, assigning the set variables as functions of the multiset ones so as to meet
the keys while preserving the separating multiplicity. Whether a computable bound exists here
remains open.

\section{Queries with Comparisons}
\label{sec:comparisons}

Throughout, $Q\leftarrow L,M$ and $Q'\leftarrow L',M'$ are combined-semantics queries with order atoms.
Until \cref{lip:sec} the schema is unconstrained and every order atom is \emph{var-const}:
$y\,\rho\,c$ with $\rho\in\{<,\le,>,\ge,=,\ne\}$, $y$ a variable and $c$ a constant.
Variable-versus-variable comparisons are set aside until \cref{lip:sec}, which readmits those
pinned by keys.
Renaming nondistinguished variables apart, we assume $Q,Q'$ share no variables. Values are drawn from a dense
linear order without endpoints, taken to be $\mathbb Q$. Density is the only property we use: it lets us
drop a fresh value into any open interval. Write $R_Q$ for the \emph{relational part} of $Q$
(delete the order atoms). Following the standing assumption, $Q$ and $Q'$ are set-equivalent, so the sole
question is whether their multiplicities agree.

The whole section turns on one idea. The constants of the two queries cut the domain into
finitely many intervals, a comparison reads a value only through the interval it lies in, and
once each value is fixed to an interval, what remains is the relational counting of
\cref{sec:bound}. The comparisons decide only which values appear, not how many atoms, so the
witness never grows beyond the comparison-free bound $2^{w}|Q|$. What varies is what fixes the
interval of a compared value in the witness. \Cref{comp:framework} builds the machinery, and
\cref{comp:certified} proves the bound whenever every position-group of the pair carries one of
three \emph{certificates}: comparisons pointing one way, values retained in the answer, or a
private column whose interval set-equivalence pins. \Cref{lip:sec} adds a fourth, values
determined by keys. The one configuration with no certificate---a two-sided comparison on a
projected-away value entangled with counted ones---is the case we leave open
(\cref{comp:frontier}).

\subsection{Slots, placements, and homomorphisms}
\label{comp:framework}

Let $\mathcal C=\{c_1<\dots<c_k\}$ collect the constants of both $Q$ and $Q'$. The constants cut
$\mathbb Q$ into the \emph{slots}: the $k$
singleton sets $\{c_i\}$ and the $k+1$ open intervals between consecutive constants, from
$(-\infty,c_1)$ to $(c_k,\infty)$. We write $\mathcal S$ for the set of slots.

An order atom $y\,\rho\,c$ holds or fails for a value depending only on which slot the value lies in,
since its threshold $c$ is a slot boundary. The \emph{region} $R_y\subseteq\mathcal S$ collects the
slots on which every order atom on $y$ holds, with $R_y=\mathcal S$ when $y$ is unconstrained. 
Throughout, we assume that no region $R_y$ is the single slot $\{c\}$ of a
constant, as an equality $y=c$ forces. Such a $y$ takes the value $c$ in every satisfying
assignment, so replacing $y$ by $c$, and dropping it from $M$ when it is there, preserves every
answer. Every region then contains at least one open slot.
Call $y$ \emph{simple} if $R_y$ is a single open half-line (one strict
comparison) or all of $\mathcal S$ (unconstrained), and \emph{composite} otherwise. For a simple $y$
with an upper bound, we write $\mathrm{ub}(y)=\min\{c:(y<c)\in Q\}$, dually $\mathrm{lb}(y)$, and call
the slot just below $\mathrm{ub}(y)$, or just above $\mathrm{lb}(y)$, the \emph{tight slot} of $y$.

\myparagraph{Placed canonical family}
Fix a placement $\pi$ assigning each variable a slot $\pi(y)\in R_y$, and block sizes
$\vec N=(N_1,\dots,N_m)\in\mathbb N_+^{m}$, one per multiset variable ($m=|M|$). The \emph{placed
family} $\pcdb{\vec N}{Q}$ repeats the blow-up of \cref{sec:cmono}, drawing every fresh value from
the slot $\pi$ dictates: the distinguished and set variables receive fixed fresh values inside their
slots, and the $j$-th multiset variable a block of $N_j$ fresh values inside the open interval
$\pi(Y_j)$, all values pairwise distinct and avoiding $\mathcal C$. (A variable placed in a singleton
slot $\{c\}$ takes the value $c$ and carries no block.) The head image $\bar a$ is the same for
every $\vec N$, and $\mathrm{mult}_Q(\mathcal D,\bar a)$ denotes the multiplicity of $\bar a$ in
$Q(\mathcal D)$, abbreviated $\mathrm{mult}_Q(\mathcal D)$ when the answer is that head image.

Two observations drive everything below. \emph{Order-invariance}: as a relational structure,
$\pcdb{\vec N}{Q}$ is the canonical family $\cdb{\vec N}{R_Q}$ of \cref{sec:cmono}, since the two differ
only in the numeric values of the fresh constants, which a relational query cannot see.
\emph{All-or-none}: every threshold of $Q'$ lies in $\mathcal C$, a slot boundary, so it never falls inside
a block. Each block therefore lies entirely inside or entirely outside any region $R_{y'}$, and likewise
each fixed value. So on $\pcdb{\vec N}{Q}$ the multiplicity of $Q'$ equals that of its relational
part $R_{Q'}$ restricted to assignments sending each variable into a block or value whose slot lies in that
variable's $Q'$-region, and $\mathrm{mult}_{Q'}(\pcdb{\vec N}{Q})$ is an integer polynomial in $\vec N$.

A multiset-homomorphism
$\mu\colon Q'\to Q$ \emph{respects regions} when $R_{\mu(y')}\subseteq R_{y'}$ for every $y'$ sent to
a variable, and $\{c\}\in R_{y'}$ for every $y'$ sent to a constant $c$.

\begin{restatable}[Soundness]{lemma}{compsound}
\label{comp:sound}
If there are region-respecting multiset-homomorphisms $Q'\to Q$ and $Q\to Q'$, then $Q\equiv Q'$.
\end{restatable}

Soundness needs no restriction on regions or variables: composing assignments with a
region-respecting homomorphism preserves the comparisons, so each direction bounds one multiplicity by
the other. The converse is the subject of the next subsection.

\subsection{The Certified Bound}
\label{comp:certified}

Comparisons on columns that never meet cannot interfere. To make this precise, form the
\emph{position graph} on the variables of $Q$ and $Q'$ together, joining two variables with an
edge whenever they occupy a common \emph{position}, the same argument slot of the same
predicate, in a relational atom of either query. Its connected components are the
\emph{position-groups}. Homomorphisms respect groups, since a containment map sends each
$p$-atom to a $p$-atom position-wise, so each variable goes to a constant or to a variable of its
own group.

Say a position-group is \emph{directed} if every comparison on its variables is strict and all
point the same way, all upper or all lower. Different groups may point differently. Call a
variable \emph{retained} if it is distinguished or multiset, and a group \emph{retained} if all
its variables are. Call a group \emph{isolated} if it is $\{y,y'\}$ with $y$ from $Q$ and $y'$
from $Q'$, and no constant occurs at a position they occupy. Being a group of two already keeps
every other variable out of those positions, and excluding constants leaves a homomorphism in
either direction no image for $y'$ but $y$, and none for $y$ but $y'$, whether the variables are
multiset or existential. Finally, call a constant $c$ \emph{boundary-free} if
it occurs in $Q,Q'$ only in weak comparisons on existential variables, never in a strict
comparison, never constraining an $(X\cup M)$-variable, never inside a relational atom.

\begin{restatable}[Certified comparison bound]{theorem}{compmaster}
\label{comp:master}
Let $Q,Q'$ be queries with var-const comparisons such that, after replacing the weak comparisons
at boundary-free constants by their strict forms, every position-group of the pair is directed,
retained, or isolated. If $Q\not\equiv Q'$, they differ on an ordered database of at most
$2^{w}|Q|$ atoms.
\end{restatable}

The proof handles each group by its certificate. We present the three
mechanisms in turn, each with a small example that calls for it.

\myparagraph{Directed groups}
A single one-directional comparison has a tight slot, and blocks want to sit in it. For $y<5$
against $y'<3$, a block just below $5$ lands in $(3,5)$, which the lower threshold cannot
follow. In general, place every variable of a
directed group in its tight slot. By
all-or-none and \cref{prop:cmono-sep}, the coefficient of the characteristic monomial in
$\mathrm{mult}_{Q'}(\pcdb{\vec N}{Q})$ then counts exactly the region-respecting homomorphisms
$Q'\to Q$, since a tight slot lands inside a region of the same direction precisely when the
regions are contained. Comparing this count with the same count for $Q$ itself yields a
characterization. Say $Q,Q'$ are \emph{simple} if every position-group of the pair is directed.

\begin{restatable}[Criterion for simple pairs]{corollary}{compmonothm}
\label{comp:mono-thm}
Let $Q,Q'$ be simple. Then $Q\equiv Q'$ if and only if there are region-respecting
multiset-homomorphisms $Q'\to Q$ and $Q\to Q'$.
\end{restatable}

\myparagraph{Retained groups}
A two-sided comparison has no tight slot. For $2<y<8$ against $2<y'<5$, a counted value may sit
in $(2,5)$, which $Q'$ can follow, or in $(5,8)$, which it cannot, so no placement suffices. We
split the range at its interior constant $5$ and treat each slot on its own. The split counts
correctly only when the compared value is retained, its slot then determined by the answer, so
answers from different slots are never merged.\looseness=-1

For a query $P$, standing for either of $Q$ and $Q'$, a \emph{resolution} $\sigma$ picks one
slot $\sigma(y)\in R_y$ for every variable $y$ of every retained group that is not directed. The
\emph{reduct} $P_\sigma$ confines each such $y$ to $\sigma(y)$: when $\sigma(y)$ is an open
slot, by the two comparisons cutting it out, and when $\sigma(y)=\{c\}$, by replacing $y$ with
$c$ outright, so that no variable is frozen at a constant. A satisfying assignment of $P$ places
each resolved variable in exactly one slot, and the resolved variables are retained, so an
assignment's $(X\cup M)$-image contributes to exactly one reduct. Hence
$\mathrm{mult}_P(\mathcal D,\bar a)=\sum_{\sigma}\mathrm{mult}_{P_\sigma}(\mathcal D,\bar a)$
for every $\mathcal D$ and $\bar a$.
Every group of a reduct is directed or pinned to a single open slot, and the exact-placement
count extends to the pinned groups. The proof of \cref{comp:master} orders the reducts of both
queries by the existence of region-respecting homomorphisms and reads a maximal class off its
own placed family.\looseness=-1

Weak comparisons are the smallest instance: $y\le c$ resolves into the strict $y<c$ and the
boundary $y=c$, so the single-atom database $\{r(0,5)\}$ separates the counted $z\le5$ from the
counted $y<5$ over $r(x,\cdot)$. When every
group is directed or retained, $Q\equiv Q'$ holds exactly when $Q$ and $Q'$ carry the same
multiset of reducts, up to equivalence. That criterion is no short certificate, the resolutions
numbering the product of the region sizes, so no counterpart of \cref{cor:np} accompanies
\cref{comp:master}.

\myparagraph{Isolated groups}
When the compared value is projected away, splitting over-counts: the value is not part of the
answer. Yet a range on such a column can still be harmless. Consider
\texttt{SELECT D.x FROM (SELECT DISTINCT r.x, r.z FROM r WHERE r.y > 2 AND r.y < 8) D}.
The column \texttt{r.y} is existential and its comparison is two-sided, so neither mechanism
above applies. But \texttt{r.y} occurs nowhere else: if a second query filtered it by a
different range, a row with a value in the gap between the ranges would break set-equivalence,
which thus pins the range even though \texttt{r.y} is projected away.

\begin{restatable}[Region pinning]{lemma}{comppin}
\label{comp:pin}
Let $y$ be isolated, paired with $y'$, confined to $R_y$ and $R_{y'}$. Under the standing set-equivalence
assumption $R_y=R_{y'}$.
\end{restatable}

The proof freezes $Q$ with $y$ placed in a slot of $R_y\setminus R_{y'}$
and derives a contradiction with set-equivalence. Once the regions match, the comparisons are
inert. Place both variables of an isolated group in an open slot of the common region. By
all-or-none, every comparison accepts the placed block or value as a whole, and a homomorphism
maps the pair to each other, the only images available. An isolated group therefore never
disturbs the count. When every comparison sits on an isolated variable, the comparisons drop out
of the criterion altogether.

\begin{restatable}[Criterion for isolated pairs]{corollary}{compisolatedthm}
\label{comp:isolated-thm}
Suppose every comparison of $Q,Q'$ is a var-const atom on an isolated variable, of either kind
and of any shape. Then $Q\equiv Q'$ if and only if their relational parts are
multiset-homomorphic.
\end{restatable}

\myparagraph{Boundary-free constants}
Last, the strictification in the theorem. Suppose a query filters an existential column by
$z\le 9$, and the constant $9$ appears nowhere else in either query. No answer changes when
every occurrence of the value $9$ in a database is lowered slightly, so counter-examples may
avoid the value $9$, and on such databases $z\le9$ acts as the strict $z<9$.
Weak comparisons at boundary-free constants may therefore be made
strict before the certificates are checked.

\begin{remark}[The frontier]
\label{comp:frontier}
What is left is a position-group with no certificate: a composite comparison on an existential
variable that is neither isolated nor at a boundary-free constant. Then splitting over-counts
and set-equivalence no longer pins the region. For instance, take
$Q(x)\leftarrow r(x,y),r(x,u),\ 0<y<1,\ 0<u<2,\ \{y\}$ and
$Q'(x)\leftarrow r(x,y'),r(x,u'),\ 0<y'<2,\ 0<u'<1,\ \{y'\}$, with $u,u'$ existential. The
queries are set-equivalent. Each states that some $r(x,\cdot)$ value lies in
$(0,1)$. Yet the regions of the multiset variables differ, $(0,1)\ne(0,2)$, and the queries are
inequivalent: on $\{r(0,\tfrac12),r(0,\tfrac32)\}$ they return multiplicities $1$ and $2$. The
counter-example is small, but none of our arguments delivers it. This case, together with
variable-versus-variable comparisons on variables that keys do not pin, where a block can be
split from inside and all-or-none fails outright, is left to future work.
\end{remark}

\subsection{Comparisons on key-determined variables}
\label{lip:sec}

The results so far leave two cases out of reach: a composite comparison on a variable that both blows up
and shares its group (\cref{comp:frontier}), and any comparison between two variables. Both become free
when the compared variables are \emph{pinned} by keys, so their values are determined by the answer and
cannot blow up. This is the ubiquitous SQL pattern of filtering on attributes fixed by keys, as in
selecting the line items sold below list price: $\texttt{I.price}<\texttt{P.list}$ relates two
values fixed by the keys of the output row and only decides which rows appear. We work over a schema with declared keys and, following the standing set-equivalence
assumption, take $Q,Q'$ set-equivalent over legal databases.

Call a variable \emph{key-anchored} if it is $X$-determined, lying in the closure $X^{+}$ of the
distinguished variables under the key dependencies of $Q$ (\cref{sec:keyed}), and a comparison atom, var-const or var-var,
\emph{pinned} if every variable it mentions is key-anchored.
Write $Q_0$ for $Q$ with all comparison atoms deleted. Without declared keys $X^{+}=X$, so pinned means
\emph{on distinguished variables}. A declared key enlarges $X^{+}$ to any attribute reached from the output
through key lookups. Two facts make pinned comparisons free. First, a pinned comparison is a \emph{support
gate}: on a legal ordered database, any two satisfying assignments of $Q_0$ producing the answer $\bar a$
agree on every key-anchored variable, so $\mathrm{mult}_Q(\mathcal D,\bar a)$ equals
$\mathrm{mult}_{Q_0}(\mathcal D,\bar a)$ when the values forced by $\bar a$ satisfy every comparison, and $0$
otherwise. Second, a comparison-free query is \emph{order-blind}: an injection $\sigma$ of values fixing the
constants of $P$ has $\mathrm{mult}_P(\sigma\mathcal D,\sigma\bar a)=\mathrm{mult}_P(\mathcal D,\bar a)$, and
$\sigma\mathcal D$ is legal whenever $\mathcal D$ is.

\begin{restatable}[Pinned comparison bound]{theorem}{lipthm}
\label{lip:thm}
Let $Q,Q'$ be set-equivalent over legal databases, with every comparison atom, var-const or
var-var, pinned, and let their comparison-free bodies $Q_0,Q'_0$, obtained by deleting all
comparison atoms, be key-anchorable. If
$Q\not\equiv Q'$, they differ on a legal ordered database of at most $2^{\mathit{kw}}|Q|$ atoms.
\end{restatable}

The proof uses the support gate to reduce inequivalence to the comparison-free
bodies, applies \cref{thm:keyed-bound} to them, and reorders the witness's values by an injection so that
the pinned comparisons hold, which order-blindness allows. 

\section{Related Work}
\label{sec:related}

\myparagraph{Equivalence characterizations}
Equivalence of conjunctive queries under set semantics is characterized by containment
mappings~\cite{ChandraM77}, under bag and bag-set semantics by isomorphism of the
queries~\cite{ChaudhuriVardi93}, and with comparisons via
linearizations~\cite{Klug,vanderMeyden,EquivAggregate}. Cohen introduced \emph{combined semantics}, subsuming
set, bag, and bag-set semantics, and characterized equivalence for several
classes---including relational queries---by
\emph{multiset-homomorphisms}~\cite{Cohen06,Cohen09}, the characterization our bound rests on
(\cref{thm:relational-equiv}). Chirkova extended the study to \emph{copy-sensitive} queries over
bag-valued relations, characterized by \emph{covering mappings} on an \emph{explicit-wave}
subclass with piecewise-polynomial multiplicity
functions~\cite{rada-icdt,rada-jcss,rada-computer-journal}. Bag-containment remains an open problem~\cite{MarcinkowskiON25}.
Equivalence under integrity constraints is classically handled by chasing the constraints into
the queries~\cite{AhoSagivUllman79}, as our foreign-key analysis does.\looseness=-1

\myparagraph{Practical equivalence checkers}
Recent systems check SQL equivalence with solvers, in the two one-sided families of
\cref{sec:intro}. Cosette pairs counter-example search with a proof-assistant
backend~\cite{Cosette}. EQUITAS translates queries to first-order formulas discharged by an SMT
solver~\cite{EQUITAS}. SQLSolver reduces the unbounded summations of bag semantics to linear
integer arithmetic~\cite{SQLSolver}. QED decides a substantial fragment via
$Q$-expressions~\cite{QED}. Each certifies equivalence on its own fragment. Only Cosette
emits counter-examples. On the refuting side, VeriEQL encodes all databases with at most $n$ tuples
per relation as an SMT formula, escalating $n$, and leads the field in integrity-constraint
support~\cite{VeriEQL}. Polygon fixes $n$ and searches under-approximations of operator
behaviour exhaustively, so an empty-handed search is definitive relative to the
bound~\cite{Polygon}. SpotIt turns the machinery into a correctness oracle for text-to-SQL
benchmarks~\cite{SpotIt}. No completeness threshold was known for any of them. Ours supply it
for the fragment covered here.\looseness=-1

\section{Conclusion}
\label{sec:conclusion}

We proved computable counter-example bounds for the equivalence of conjunctive queries under
combined semantics: $2^{w}|Q|$ over unconstrained set databases, $2^{\mathit{kw}}|Q|$ under
declared keys, unchanged under acyclic foreign keys, and $2^{w}|Q|$ again for several
comparison classes previously lacking any equivalence characterization. Up to these thresholds,
bounded search becomes a terminating, complete proof method.\looseness=-1

Three problems remain open. First, the keyed bound requires key-anchorability, leaving open the
queries in which demotion strands a multiset variable off every key. Second, the comparison
frontier: a two-sided comparison on a projected-away value neither isolated nor key-determined,
and var-var comparisons not pinned by keys. Third, \emph{tightness}: the corner $\{1,2\}^m$
cannot be relaxed, yet on the pair forcing it the size bound is loose. We know no family forcing
size $2^{\Omega(w)}$, and whether the true threshold is polynomial in the query is open. Beyond
these, extending the bounds to bag-valued stored relations~\cite{rada-icdt,rada-jcss} and to
further query classes, aggregation foremost, is a natural next step.\looseness=-1

\bibliography{references}

\begin{thebibliography}{10}

\bibitem{AhoSagivUllman79}
Alfred~V. Aho, Yehoshua Sagiv, and Jeffrey~D. Ullman.
\newblock Equivalences among relational expressions.
\newblock {\em SIAM Journal on Computing}, 8(2):218--246, 1979.
\newblock \href {https://doi.org/10.1137/0208017} {\path{doi:10.1137/0208017}}.

\bibitem{Alon99}
Noga Alon.
\newblock Combinatorial nullstellensatz.
\newblock {\em Combinatorics, Probability and Computing}, 8(1--2):7--29, 1999.
\newblock \href {https://doi.org/10.1017/S0963548398003411}
  {\path{doi:10.1017/S0963548398003411}}.

\bibitem{ChandraM77}
Ashok~K. Chandra and Philip~M. Merlin.
\newblock Optimal implementation of conjunctive queries in relational data
  bases.
\newblock In {\em Proceedings of the Ninth Annual ACM Symposium on Theory of
  Computing}, STOC '77, page 77–90, New York, NY, USA, 1977. Association for
  Computing Machinery.
\newblock \href {https://doi.org/10.1145/800105.803397}
  {\path{doi:10.1145/800105.803397}}.

\bibitem{ChaudhuriVardi93}
Surajit Chaudhuri and Moshe~Y. Vardi.
\newblock Optimization of real conjunctive queries.
\newblock In {\em Proceedings of the Twelfth ACM SIGACT-SIGMOD-SIGART Symposium
  on Principles of Database Systems}, PODS '93, page 59–70, New York, NY,
  USA, 1993. Association for Computing Machinery.
\newblock \href {https://doi.org/10.1145/153850.153856}
  {\path{doi:10.1145/153850.153856}}.

\bibitem{rada-icdt}
Rada Chirkova.
\newblock Equivalence and minimization of conjunctive queries under combined
  semantics.
\newblock In {\em Proceedings of the 15th International Conference on Database
  Theory}, ICDT '12, page 262–273, New York, NY, USA, 2012. Association for
  Computing Machinery.
\newblock \href {https://doi.org/10.1145/2274576.2274604}
  {\path{doi:10.1145/2274576.2274604}}.

\bibitem{rada-computer-journal}
Rada Chirkova.
\newblock Combined-semantics equivalence and minimization of conjunctive
  queries.
\newblock {\em Comput. J.}, 57(5):775--795, 2014.
\newblock URL: \url{https://doi.org/10.1093/comjnl/bxt032}, \href
  {https://doi.org/10.1093/COMJNL/BXT032} {\path{doi:10.1093/COMJNL/BXT032}}.

\bibitem{rada-jcss}
Rada Chirkova.
\newblock Combined-semantics equivalence of conjunctive queries: Decidability
  and tractability results.
\newblock {\em Journal of Computer and System Sciences}, 82(3):395--465, 2016.
\newblock URL:
  \url{https://www.sciencedirect.com/science/article/pii/S0022000015001129},
  \href {https://doi.org/10.1016/j.jcss.2015.11.001}
  {\path{doi:10.1016/j.jcss.2015.11.001}}.

\bibitem{Cosette}
Shumo Chu, Chenglong Wang, Konstantin Weitz, and Alvin Cheung.
\newblock Cosette: An automated prover for {SQL}.
\newblock In {\em 8th Biennial Conference on Innovative Data Systems Research
  ({CIDR})}, 2017.

\bibitem{Cohen06}
Sara Cohen.
\newblock Equivalence of queries combining set and bag-set semantics.
\newblock In {\em Proceedings of the Twenty-Fifth ACM SIGMOD-SIGACT-SIGART
  Symposium on Principles of Database Systems}, PODS '06, pages 70--79, New
  York, NY, USA, 2006. Association for Computing Machinery.
\newblock \href {https://doi.org/10.1145/1142351.1142362}
  {\path{doi:10.1145/1142351.1142362}}.

\bibitem{Cohen09}
Sara Cohen.
\newblock Equivalence of queries that are sensitive to multiplicities.
\newblock {\em {VLDB} J.}, 18(3):765--785, 2009.
\newblock URL: \url{https://doi.org/10.1007/s00778-008-0122-1}, \href
  {https://doi.org/10.1007/S00778-008-0122-1}
  {\path{doi:10.1007/S00778-008-0122-1}}.

\bibitem{EquivAggregate}
Sara Cohen, Werner Nutt, and Yehoshua Sagiv.
\newblock Deciding equivalences among conjunctive aggregate queries.
\newblock {\em J. ACM}, 54(2):5–es, April 2007.
\newblock \href {https://doi.org/10.1145/1219092.1219093}
  {\path{doi:10.1145/1219092.1219093}}.

\bibitem{SQLSolver}
Haoran Ding, Zhaoguo Wang, Yicun Yang, Dexin Zhang, Zhenglin Xu, Haibo Chen,
  Ruzica Piskac, and Jinyang Li.
\newblock Proving query equivalence using linear integer arithmetic.
\newblock {\em Proceedings of the ACM on Management of Data},
  1(4):227:1--227:26, 2023.
\newblock \href {https://doi.org/10.1145/3626768} {\path{doi:10.1145/3626768}}.

\bibitem{VeriEQL}
Yang He, Pinhan Zhao, Xinyu Wang, and Yuepeng Wang.
\newblock {VeriEQL}: Bounded equivalence verification for complex {SQL} queries
  with integrity constraints.
\newblock {\em Proc. {ACM} Program. Lang.}, 8({OOPSLA1}):132:1--132:29, 2024.
\newblock \href {https://doi.org/10.1145/3649849} {\path{doi:10.1145/3649849}}.

\bibitem{SpotIt}
Rocky Klopfenstein, Yang He, Andrew Tremante, Yuepeng Wang, Nina Narodytska,
  and Haoze Wu.
\newblock Spotit: Evaluating text-to-{SQL} evaluation with formal verification.
\newblock In {\em The Fourteenth International Conference on Learning
  Representations, {ICLR} 2026}, 2026.
\newblock \href{https://arxiv.org/abs/2510.26840}{arXiv:2510.26840}.

\bibitem{Klug}
Anthony Klug.
\newblock On conjunctive queries containing inequalities.
\newblock {\em Journal of the ACM}, 35(1):146--160, 1988.
\newblock \href {https://doi.org/10.1145/42267.42273}
  {\path{doi:10.1145/42267.42273}}.

\bibitem{MarcinkowskiON25}
Jerzy Marcinkowski and Piotr Ostropolski{-}Nalewaja.
\newblock Bag semantics query containment: The {CQ} vs. {UCQ} case and other
  stories.
\newblock {\em Proceedings of the ACM on Management of Data},
  3(5):275:1--275:24, 2025.
\newblock \href {https://doi.org/10.1145/3767711} {\path{doi:10.1145/3767711}}.

\bibitem{QOVerify}
Vivek~R. Narasayya and Surajit Chaudhuri.
\newblock Leveraging query optimizers to verify the soundness of {LLM}-based
  query rewrites for real-world workloads, and more.
\newblock In {\em 16th Conference on Innovative Data Systems Research, {CIDR}
  2026}, 2026.

\bibitem{ullman1988principles2}
Jeffrey~D. Ullman.
\newblock {\em Principles of Database and Knowledge-Base Systems, Volume II}.
\newblock Computer Science Press, 1988.

\bibitem{vanderMeyden}
Ron van~der Meyden.
\newblock The complexity of querying indefinite data about linearly ordered
  domains.
\newblock {\em Journal of Computer and System Sciences}, 54(1):113--135, 1997.
\newblock \href {https://doi.org/10.1006/jcss.1997.1455}
  {\path{doi:10.1006/jcss.1997.1455}}.

\bibitem{QED}
Shuxian Wang, Sicheng Pan, and Alvin Cheung.
\newblock {QED}: {A} powerful query equivalence decider for {SQL}.
\newblock {\em Proc. {VLDB} Endow.}, 17(11):3602--3614, 2024.
\newblock \href {https://doi.org/10.14778/3681954.3682024}
  {\path{doi:10.14778/3681954.3682024}}.

\bibitem{Polygon}
Pinhan Zhao, Yuepeng Wang, and Xinyu Wang.
\newblock Polygon: Symbolic reasoning for {SQL} using conflict-driven
  under-approximation search.
\newblock {\em Proceedings of the ACM on Programming Languages},
  9({PLDI}):1315--1340, 2025.
\newblock \href {https://doi.org/10.1145/3729303} {\path{doi:10.1145/3729303}}.

\bibitem{EQUITAS}
Qi~Zhou, Joy Arulraj, Shamkant~B. Navathe, William Harris, and Dong Xu.
\newblock Automated verification of query equivalence using satisfiability
  modulo theories.
\newblock {\em Proc. {VLDB} Endow.}, 12(11):1276--1288, 2019.
\newblock \href {https://doi.org/10.14778/3342263.3342267}
  {\path{doi:10.14778/3342263.3342267}}.

\end{thebibliography}

\appendix
\numberwithin{theorem}{section}
\makeatletter
\@for\@tempa:={theorem,lemma,corollary,proposition,definition,observation,example,remark,claim,conjecture}\do{%
  \expandafter\edef\csname the\@tempa\endcsname{\noexpand\thesection.\noexpand\arabic{theorem}}}
\makeatother
\crefalias{section}{appendix}
\crefalias{subsection}{appendix}
\clearpage
\section*{Use of AI Tools}

We disclose our use of AI tools, following the ACM Policy on Authorship. A general-purpose AI
coding assistant, Anthropic's Claude Code, was used throughout the preparation of this work. Its use went beyond
assistance with the writing, so we describe it here in full.

The assistant was used in three ways that bear on the results. First, it wrote Python scripts
that search for small counter-examples and evaluate candidate queries. We used those scripts to
probe conjectures before attempting to prove them, and to look for refutations of statements we
suspected were false. No claim in this paper rests on them. The paper reports no experiments,
and every statement we kept was proved afterwards. Second, it expanded human-written proof
sketches into the full proofs given in the appendices. The proof strategy, the constructions,
and the case analyses are ours. Third, it pointed out connections between our setting and
existing mathematical tools. At least one such connection is used in the paper, and the result
it rests on is cited where it is applied.

The authors read and checked every definition, statement, proof, and reference in this paper,
and take full responsibility for all of them. AI tools are not authors and made no decision
about what this paper claims.

\section{Formalization of the Canonical Construction}
\label{app:construction}

This appendix makes \cref{sec:bound,sec:keyed} self-contained by repeating the key constructions and proofs
from~\cite{Cohen09}. A reader can verify the size, degree, and
separation claims on which the counter-example bound rests.

Throughout, $Q\leftarrow L,M$ is a relational query with $m=|M|$ multiset variables. We assume
$m\ge 1$. The case $M=\emptyset$ is the classical containment-mapping theorem~\cite{ChandraM77} and
needs no construction. We name the variables of $Q$ as follows: $X_1,\dots,X_{l}$ are the head
variables, $X_{l+1},\dots,X_{l+u}$ the nonhead set variables, and $Y_1,\dots,Y_m$ the multiset
variables.

\begin{definition}[The canonical database family $\cdb{\vec N}{Q}$, \cite{Cohen09}]
\label{app:def-family}
Let $\nu_0$ be an injection assigning a value to each head variable, each set variable, and each constant
of $Q$, fixed once and for all. It maps each constant to itself, and all other values are fresh,
pairwise distinct, and distinct from every constant occurring in $Q$ or $Q'$. Write $\bar c := \nu_0(\bar x)$ for the image of the head, the \emph{distinguished
tuple}. It is the same in every database below. Given
$\vec N=(N_1,\dots,N_{m})\in\mathbb{N}_+^{m}$, construct $\cdb{\vec N}{Q}$ as follows.
\begin{enumerate}
\item For each $j\in\{1,\dots,m\}$ pick a set $D_j$ of $N_j$ fresh constants, with the $D_j$
pairwise disjoint, disjoint from the range of $\nu_0$, and disjoint from every constant occurring in
$Q$ or $Q'$. Let $D := D_1\times\cdots\times D_m$ (a single empty tuple if $m=0$).
\item For each $\bar d=(d_1,\dots,d_m)\in D$, let $\nu_{\bar d}$ be the assignment agreeing with
$\nu_0$ on the head and set variables and on constants, and sending each $Y_j\mapsto d_j$.
\item \emph{Main loop.} For each $\bar d\in D$ and each subgoal $p(\bar s)\in Q$, add the ground
atom $p(\nu_{\bar d}(\bar s))$. The database $\cdb{\vec N}{Q}$ consists of exactly these atoms
(identical atoms merged, so the result is a set).
\end{enumerate}
\end{definition}

Each tuple $\bar d\in D$ gives a satisfying assignment of $Q$ into $\cdb{\vec N}{Q}$ taking the
head to $\bar c$, so $\bar c\in Q(\cdb{\vec N}{Q})$ for every $\vec N$. The main loop iterates over the
$\prod_{i=1}^m N_i$ tuples of $D$ and adds at most $|Q|$ atoms per tuple. The size of each relation
after merging is bounded in \cref{lem:corner-size}.

\begin{example}
\label{app:ex-construction}
Take $Q(x)\leftarrow r(x,y)\wedge s(y,z),\ \{y\}$, with $y$ the single multiset variable and $z$ a set
variable, so $m=1$ and $|Q|=2$. The fixed values send $x\mapsto\hat x$ and $z\mapsto\hat z$ (fresh and
distinct), giving $\bar c=(\hat x)$. For a parameter $N_1$, blow up $y$ over
$D_1=\{a_1,\dots,a_{N_1}\}$. The main loop then inserts, for each $a_i$, the atoms $r(\hat x,a_i)$
and $s(a_i,\hat z)$, so
\[
\cdb{(N_1)}{Q}=\{\,r(\hat x,a_i):1\le i\le N_1\,\}\,\cup\,\{\,s(a_i,\hat z):1\le i\le N_1\,\},
\]
with $r$ and $s$ each of size $N_1$ ($2N_1$ atoms in all). Each $a_i$ gives one satisfying assignment with
head $\hat x$, taking $z=\hat z$, and these are all of them. So $\bar c$ has multiplicity $N_1$, exactly
the characteristic monomial $\mathcal P^{(Q)}_*=N_1$.
\end{example}

Fix a relational query $R$. The construction expresses the multiplicity of $\bar c$ in
$R(\cdb{\vec N}{Q})$ as a function of $\vec N$: write $\mathcal{F}^{(R)}(\vec N)$ for the multiplicity of
$\bar c$ in the bag $R(\cdb{\vec N}{Q})$, that is, the number of assignments in
$\Gamma_{X_R\cup M_R}(R,\cdb{\vec N}{Q})$ whose head image is $\bar c$.

To describe $\mathcal{F}^{(R)}$, \cite{Cohen09} classifies the $\bar c$-producing assignments of $R$.
Each ground atom of $\cdb{\vec N}{Q}$ has a unique subgoal of $Q$ as its template, namely its
$\nu_Q$-image. An \emph{association} records, for one such assignment, the template of each subgoal of
$R$. It also records a \emph{signature}: for each multiset variable of $R$, which multiset variable $Y_i$
of $Q$, or which fixed value of $\nu_0$, its image lies under.

Now group the $\bar c$-producing assignments by association. Suppose an association's signature sends the
multiset variables of $R$ into domains $D_{i_1},\dots,D_{i_k}$, with repetition. It is realized by exactly
$\prod_{j} N_{i_j}$ assignments, a \emph{monomial}. Its total degree is the number of multiset variables
of $R$ mapped into blown-up domains, hence at most $m_R$. Distinct associations contribute additively. Thus
$\mathcal{F}^{(R)}$ is the sum of these monomials.

Among the monomials of $\mathcal{F}^{(Q)}$ one is distinguished: the \emph{characteristic monomial}
$\mathcal{P}^{(Q)}_* = N_1 N_2\cdots N_m$ of $Q$. It has total degree $m$ and contains every parameter
$N_1,\dots,N_m$.
The three facts below, the black boxes used in \cref{sec:bound}, summarize what we need of these
multiplicity functions.

\begin{proposition}[Multiplicity polynomials and the characteristic monomial; \cite{Cohen09}, Lemma~5.2 and Thm.~5.3]
\label{app:sep-facts}
Let $Q,Q'$ be relational queries and $R$ \emph{any} relational query, with the family $\cdb{\vec N}{Q}$
built from $Q$.
\begin{enumerate}
\item $\mathcal{F}^{(R)}$ is an integer polynomial in $N_1,\dots,N_m$, and every monomial of it has
total degree at most $m_R=|M_R|$.
\item The characteristic monomial $\mathcal{P}^{(Q)}_*$ occurs in $\mathcal{F}^{(Q)}$ with a positive
coefficient.
\item $\mathcal{P}^{(Q)}_*$ occurs in $\mathcal{F}^{(Q')}$ if and only if there is a
multiset-homomorphism from $Q'$ to $Q$.
\end{enumerate}
\end{proposition}

\begin{proof}[Proof sketch]
(1)~Each multiset variable of $R$ contributes at most one parameter factor, so every monomial above has
total degree at most $m_R$; integrality and the polynomial form follow from the additive count over
associations.
(2)~The identity assignment of $Q$ places each $Y_j$ in its own domain $D_j$, giving an association with
signature $(Y_1,\dots,Y_m)$ that contributes $N_1\cdots N_m$. It is the unique association of degree $m$
using every domain exactly once, so nothing cancels it and its coefficient is positive.
(3)~A monomial $N_1\cdots N_m$ in $\mathcal{F}^{(Q')}$ comes from an association of a special form. Its
signature is a bijection from the multiset variables of $Q'$ onto $Y_1,\dots,Y_m$, and its templates send
every subgoal of $Q'$ to a subgoal of $Q$. Read as a map on terms, such an association is a
multiset-homomorphism from $Q'$ to $Q$: it fixes the head and constants, maps atoms to atoms, and is
injective on the multiset variables. Conversely, such a homomorphism induces the association, hence the
monomial. See~\cite{Cohen09}.
\end{proof}

These facts are what the separation \cref{prop:cmono-sep} invokes. With them, the separation argument and
the size accounting of \cref{lem:corner-size} give the bound theorems of \cref{sec:bound,sec:keyed} entirely by the
(new) reasoning presented there.

\section{Omitted Proofs}
\label{app:bound-proofs}

\subsection{Unconstrained Schemas}
\label{app:unconstrained}

\cmonosep*

\begin{proof}
By \cref{app:sep-facts}, $\mathcal{P}^{(Q)}_*$ is absent from $\mathcal{F}^{(Q')}$, as no
multiset-homomorphism from $Q'$ to $Q$ exists, yet present in $\mathcal{F}^{(Q)}$ with positive
coefficient. Its coefficient in $P$ is therefore positive, so $P\neq 0$. For the degree,
$\mathcal{F}^{(Q)}$ has total degree $m$ and $\mathcal{F}^{(Q')}$ at most $m_{Q'}\le m$.
\end{proof}

\cornersize*

\begin{proof}
The atoms a subgoal $p(\bar s)$ contributes depend on the blow-up assignment only through the multiset
variables occurring in $\bar s$, the other positions being fixed. Each such variable ranges over
$a_i\le 2$ values, so the subgoal yields at most $2^{\,\nmv{p(\bar s)}}\le 2^{w(p)}$ distinct atoms.
Summing over the $\mathit{sj}(p)$ subgoals with predicate $p$ bounds $\mathit{size}(p,\cdb{\vec a}{Q})$,
and summing over all subgoals, using $\sum_p\mathit{sj}(p)=|Q|$, bounds the total by $2^{w}|Q|$.
\end{proof}

\mainbound*

\begin{proof}
By \cref{thm:relational-equiv}, some direction admits no multiset-homomorphism. The conclusion is
symmetric in $Q$ and $Q'$, so we may assume $m_{Q'}\le m_Q$ and that no multiset-homomorphism
$Q'\to Q$ exists. Indeed, when $m_Q\ne m_{Q'}$, put the query with fewer multiset variables as $Q'$:
the characteristic monomial of $Q$ then has degree $m_Q$, exceeding the degree of every monomial of
$\mathcal F^{(Q')}$ (\cref{app:sep-facts}(1)), so no multiset-homomorphism $Q'\to Q$ exists by the
criterion (\cref{app:sep-facts}(3)). When $m_Q=m_{Q'}$, name the failing direction $Q'\to Q$.

Write $m:=m_Q$ and $P := \mathcal{F}^{(Q)} - \mathcal{F}^{(Q')}$. The Combinatorial
Nullstellensatz~\cite{Alon99} states that if $P$ has a nonzero coefficient on a monomial
$\prod_i N_i^{t_i}$ with $\sum_i t_i=\deg P$, then for any sets $A_1,\dots,A_m$ with $|A_i|>t_i$ there
is a point $\vec a\in A_1\times\cdots\times A_m$ with $P(\vec a)\neq 0$. By \cref{prop:cmono-sep}, the
coefficient of $N_1\cdots N_m$ in $P$ is nonzero and $\deg P=m$, so this monomial is a top-degree term
with every $t_i=1$. Taking each $A_i=\{1,2\}$ yields $\vec a\in\{1,2\}^m$ with $P(\vec a)\neq 0$, so
the multiplicity of $\bar c$ differs between $Q$ and $Q'$ on $\mathcal D := \cdb{\vec a}{Q}$. The size
bounds are \cref{lem:corner-size} at $\vec a$, weakened from $w(Q,p),\mathit{sj}(Q,p)$ to the
pair-maxima in the statement. The final claim is the contrapositive, as the witness is $b$-bound.
\end{proof}

A variable sitting in an unobserved attribute occurs exactly once, not in the head and in no order
atom. An unobserved multiset variable feeds the multiplicity but is otherwise inert, and several in one
atom can be merged: $k$ unobserved multiset variables ranging over domains $D_1,\dots,D_k$ contribute
exactly the multiplicity of a single multiset variable over $D_1\times\cdots\times D_k$. (For the
relational queries of \cref{sec:bound} the clause about compared variables is vacuous. It is stated so
that the notion remains correct once comparisons are present.)

The \emph{collapse} of the pair $Q,Q'$ rewrites each relation $r$ according to its unobserved
attributes. If some atom of $r$, in $Q$ or $Q'$, places a multiset variable in an unobserved position,
append one fresh attribute at the end of $r$ and delete the unobserved ones. An atom $r(\bar s)$ then
becomes $r^{\downarrow}(\bar s|_{\mathrm{obs}},\,z)$, whose new last position holds a fresh multiset
variable $z$ when $\bar s$ had a multiset variable in an unobserved position, and a fresh set variable
otherwise. If no atom of $r$ places a multiset variable in an unobserved position, $r$ is left
unchanged. The unobserved attributes are common to $Q$ and $Q'$, so the same rewriting applies to both,
yielding $Q^{\downarrow}$ and $Q'^{\downarrow}$ over a shared schema. A collapsed atom carries only the
join multiset variables of the original, plus at most one new multiset variable, so
$w(Q^{\downarrow})\le w(Q)$, and clearly $|Q^{\downarrow}|=|Q|$ and
$\mathit{sj}(Q^{\downarrow})=\mathit{sj}(Q)$. The collapse acts on databases by the same merging:
$\mathcal D$ becomes $\mathcal D^{\downarrow}$ by sending each ground atom $r(\bar c)$ of a rewritten
relation to $r^{\downarrow}(\bar c|_{\mathrm{obs}},\,\langle\bar c|_{\mathrm{unobs}}\rangle)$, the new
attribute recording the whole tuple of merged values. This map is a bijection on databases and
preserves every relation's size.

\collapselem*

\begin{proof}
The collapse preserves every answer multiplicity. A multiset variable in an unobserved position
contributed a factor equal to its domain size, and the new variable $z$, ranging over the product of
those domains, reproduces exactly that factor, while set variables there contribute nothing. So
$Q(\mathcal D)=Q^{\downarrow}(\mathcal D^{\downarrow})$ for every $\mathcal D$, and likewise for $Q'$,
which with the bijection $\mathcal D\mapsto\mathcal D^{\downarrow}$ gives the claim. (An unobserved
attribute holding a constant is simply retained. It contributes nothing to the width.)
\end{proof}

Under declared keys the collapse requires one adjustment: the schema surgery must merge a
relation's unobserved key attributes separately from its unobserved non-key attributes, so that
the key is preserved and the bijection maps legal databases to legal databases. Merging
unobserved key attributes then lowers the key-width, replacing $\mathit{kw}$ in
\cref{thm:keyed-bound} by its collapsed counterpart.

\subsection{Constrained Schemas}
\label{app:keys}

The surviving multiset variables of a run of demotions form a \emph{basis} of $M$ under the
closure $T\mapsto(X\cup T)^{+}$, and because functional-dependency closures need not obey the
matroid exchange axiom, bases are not unique and may differ in size. One might hope that the
dependencies induced by primary keys are better behaved than general functional dependencies.
They are not. Take the atoms
$p(a,b),q(a,c)$ of key $\{1\}$ and $r(b,c,a)$ of key $\{1,2\}$, with $M=\{a,b,c\}$ and
$X=\emptyset$. The induced dependencies are purely key-driven, $a\to b$, $a\to c$, and
$\{b,c\}\to a$, and the last two make $a$ and $\{b,c\}$ \emph{two candidate keys for the same
information}: $a$ determines $\{b,c\}$ and $\{b,c\}$ determines $a$. Demoting $b$ then $c$ leaves
the basis $\{a\}$, whereas demoting $a$ leaves $\{b,c\}$---bases of different size, and even of
different key-width ($1$ versus $2$). One might expect order-irrelevance once every dependency is
a key. It fails, and the culprit is precisely the cyclic mutual determination
$a\leftrightarrow\{b,c\}$, the same cyclic-key phenomenon as in \cref{ex:cyclic-key}. With
acyclic key dependencies the basis is unique.

The order affects how \emph{tight} the witness is, as a smaller basis lowers the effective
exponent. It never affects whether the reduct is key-anchored: the statically anchored variables
are never demotable and survive every run, so by \cref{prop:anchored} anchoredness is the
order-independent condition $M\subseteq(X\cup A)^{+}$, and when it holds every order reaches the
same anchored reduct $A$.

\demotionlem*

\begin{proof}
On a legal database $\mathcal D$ the key dependencies hold, so in every satisfying
assignment $\gamma$ the value $\gamma(v)$ is determined by the restriction of $\gamma$ to
$X\cup(M\setminus\{v\})$. Hence the map $\Gamma_{X\cup M}(Q,\mathcal D)\to\Gamma_{X\cup(M\setminus\{v\})}(Q,\mathcal D)$
that forgets $v$ is a bijection, and the answer multisets coincide on every
legal $\mathcal D$. Moving $v$ out of $M$ cannot raise the number of multiset variables in any key
position.
\end{proof}

Recall that a query $R$ is \emph{key-anchored} if every multiset variable of $R$ occurs, in every atom
of $R$ containing it, at a key position of that atom. The \emph{key-chase} of $R$ repeatedly unifies
the terms of two atoms $p(\bar s),p(\bar t)$ of the same predicate that agree on $\mathit{key}(p)$ (an
equality-generating chase under the key dependencies) and merges duplicates. It terminates, preserves
$\equiv$, and leaves distinct atoms of each predicate with distinct key-term tuples.

\begin{lemma}[Legality of the canonical family, anchored case]
\label{lem:legal-canonical}
If $R$ is key-chased and key-anchored, then $\cdb{\vec a}{R}$ is legal for every
corner $\vec a\in\{1,2\}^{m}$.
\end{lemma}

\begin{proof}[Proof of \cref{lem:legal-canonical}]
Let $p(\bar u),p(\bar v)\in\cdb{\vec a}{R}$ agree on $\mathit{key}(p)$. They arise from subgoals
$p(\bar s),p(\bar s')$ and assignments $\nu_t,\nu_{t'}$ ($t,t'\in S$), with $\bar u=\nu_t(\bar s)$ and
$\bar v=\nu_{t'}(\bar s')$. Recall from \cref{app:construction} that $\nu_0$ is injective on the
set/distinguished variables and constants, and the multiset domains $S_j$ are pairwise disjoint and
disjoint from $\mathrm{range}(\nu_0)$. Hence two terms receive equal values under the $\nu$'s iff they
are the same variable or the same constant.

\emph{Same subgoal} ($\bar s=\bar s'$). For each $i\in\mathit{key}(p)$, $\nu_t(s_i)=\nu_{t'}(s_i)$.
At a fixed term this is automatic, and at a multiset variable $y$ it forces $t,t'$ to agree on the
coordinate of $y$. By key-anchoredness every multiset variable of $\bar s$ occurs at some key position,
so $t,t'$ agree on all multiset variables of $\bar s$. The other terms being fixed, $\nu_t(\bar s)=\nu_{t'}(\bar s)$,
i.e.\ $\bar u=\bar v$.

\emph{Different subgoals} ($\bar s\neq\bar s'$). After the key-chase the key-term tuples
$\bar s|_{\mathit{key}(p)}$ and $\bar s'|_{\mathit{key}(p)}$ are distinct, hence differ at some key
position $i$ where $s_i,s'_i$ are different variables, different constants, or a variable and a
constant. In every case $\nu_t(s_i)\neq\nu_{t'}(s'_i)$ by the disjointness above, so $\bar u,\bar v$
disagree on $\mathit{key}(p)$, a contradiction.

Thus atoms agreeing on the key are equal, and $\cdb{\vec a}{R}$ is legal.
\end{proof}

\begin{lemma}[Distinct legal atoms at a corner]
\label{lem:keyed-corner-size}
Let $\mathcal D$ be any legal database all of whose atoms arise, as in
$\cdb{\vec a}{\cdot}$, by assigning each multiset variable one of at most two values. Then for
every predicate $p$,
\[
\mathit{size}(p,\mathcal D)\;\le\;2^{\mathit{kw}(p)}\,\mathit{sj}(p),
\qquad\text{hence}\qquad
|\mathcal D|\;\le\;2^{\mathit{kw}}\,|Q| .
\]
\end{lemma}

\begin{proof}[Proof of \cref{lem:keyed-corner-size}]
By legality, distinct atoms of $p$ differ on $\mathit{key}(p)$. A key position holds either a fixed
term or a multiset variable with at most two values, so a subgoal with predicate $p$ admits at most
$2^{\mathit{kw}(p)}$ distinct key projections, hence at most that many distinct atoms. Summing over
subgoals as in \cref{lem:corner-size} gives both bounds.
\end{proof}

\begin{proposition}[When an anchored reduct exists]
\label{prop:anchored}
Let $A$ be the set of multiset variables of $Q$ that occur only at key positions (the \emph{statically
anchored} variables). Every variable in $A$ belongs to every key-reduct of $Q$, and the following are
equivalent:
\begin{enumerate}
\item some key-reduct of $Q$ is key-anchored,
\item every key-reduct of $Q$ is key-anchored,
\item every multiset variable of $Q$ is $(X\cup A)$-determined, that is $M\subseteq(X\cup A)^{+}$.
\end{enumerate}
When these hold, $A$ is the \emph{unique} key-reduct, reached by every demotion order. The condition
holds in particular when no multiset variable occurs off a key (then $A=M$), but it is strictly
weaker: it asks only that the anchored variables \emph{determine} the rest.
\end{proposition}

\begin{proof}[Proof of \cref{prop:anchored}]
A statically anchored $v$ never occupies a non-key position, so it appears in a key dependency
only within the left-hand (key) side. Hence $v\in(X\cup T)^{+}$ implies $v\in X\cup T$, so
$v$ is never demotable and lies in every reduct.

$(3)\Rightarrow(2)$: assume $M\subseteq(X\cup A)^{+}$ and let $F$ be any reduct and $u\notin A$ a
non-anchored multiset variable. As $A\subseteq F$ for the reason just given, $A\subseteq F\setminus\{u\}$,
whence $u\in M\subseteq(X\cup A)^{+}\subseteq(X\cup(F\setminus\{u\}))^{+}$. In a reduct no surviving
multiset variable is determined by the others, so $u\notin F$. Hence no non-anchored variable survives,
$F=A$, and $A$ is key-anchored. $(2)\Rightarrow(1)$ is immediate. $(1)\Rightarrow(3)$: a reduct $F$
generates $M$, i.e.\ $M\subseteq(X\cup F)^{+}$. If $F$ is key-anchored then $F\subseteq A$, so
$M\subseteq(X\cup A)^{+}$.
\end{proof}

\wfanchored*

\begin{proof}
Let $G$ be the digraph on the counting atoms of $Q$ whose edges $c\to b$ are the references: some
multiset variable at a key position of $b$ occurs at a non-key position of $c$. By assumption $G$ is
acyclic. For a counting atom $b$, let $h(b)$ be the length of the longest directed path in $G$ ending
at $b$.

We show by induction on $h(b)$ that every variable at a key position of $b$ lies in $(X\cup A)^{+}$.
Let $v$ be such a variable. A counting atom has no set variables, and constants and distinguished
variables are immediate, so assume $v\in M$. If no occurrence of $v$ in $Q$ sits at a non-key
position, then $v\in A$. Otherwise some occurrence of $v$ sits at a non-key position, and by the first
requirement it sits in a counting atom $c$. Then $c\to b$ is an edge of $G$, so $h(c)<h(b)$. In
particular this case cannot occur when $h(b)=0$. By induction every key variable of $c$ lies in
$(X\cup A)^{+}$, and the key of $c$ determines its remaining positions, so $v\in(X\cup A)^{+}$.

Now let $x\in M$. By well-formedness $x$ occurs in some counting atom $b$. Every key variable of $b$
lies in $(X\cup A)^{+}$, and the key of $b$ determines all its positions, so $x\in(X\cup A)^{+}$.
Hence $M\subseteq(X\cup A)^{+}$, and \cref{prop:anchored} makes every key-reduct of $Q$ key-anchored.
\end{proof}

\keyedbound*

\begin{proof}
Write $Q^{\natural}$ for the key-chase of the key-reduct $Q^{\flat}$, and likewise for $Q'$. By
\cref{lem:demotion} and the key-chase, $Q\equiv Q^{\natural}$ and $Q'\equiv Q'^{\natural}$ over
legal databases. Chasing after demoting is sound: the key-chase unifies only terms at non-key
positions, since the two atoms already agree on the key, while under key-anchorability the
surviving multiset variables are those of $A$, which occupy key positions only
(\cref{prop:anchored}). So $Q^{\natural}$ is key-chased and key-anchored, and no merge raises
$\mathit{kw}$.

\emph{Clause~(1), sufficiency.} If $Q^{\natural}$ and $Q'^{\natural}$ are multiset-homomorphic
they are equivalent over all databases by \cref{thm:relational-equiv}, hence over the legal ones,
and $Q\equiv Q'$ follows.

\emph{Clause~(1), necessity, and clause~(2).} Suppose no multiset-homomorphism exists in some
direction. As in the proof of \cref{thm:main-bound} the queries may be named so that the failing
direction is $Q'^{\natural}\to Q^{\natural}$ with $m_{Q'^{\natural}}\le m_{Q^{\natural}}$. By
\cref{prop:cmono-sep} and the Boolean-corner extraction of \cref{thm:main-bound} the two
multiplicity functions differ at some corner $\vec a\in\{1,2\}^{m}$ of
$\cdb{\vec a}{Q^{\natural}}$, which is legal by \cref{lem:legal-canonical}. Hence
$Q\not\equiv Q'$ over legal databases, which is the contrapositive of necessity, and the
separating database is legal. Its size obeys \cref{lem:keyed-corner-size}, with the relevant
key-width for predicate $p$ at most $\mathit{kw}(p)$, giving clause~(2).
\end{proof}

\myparagraph{Foreign keys}
\phantomsection\label{app:fk}

Formally, the \emph{foreign-key chase} $\mathit{chase}_{\mathrm{fk}}(Q)$ repeatedly picks a
foreign key $(p,\bar\imath,q)$ and an atom $p(\bar s)\in Q$ with no atom $q(\bar t)\in Q$
satisfying $\bar s|_{\bar\imath}=\bar t|_{\mathit{key}(q)}$, and adds a fresh atom $q(\bar t)$
whose key positions copy $\bar s|_{\bar\imath}$ and whose remaining positions are pairwise fresh
\emph{set} variables.

\begin{lemma}[Foreign-key chase]
\label{lem:fk-chase}
For acyclic foreign keys, $\mathit{chase}_{\mathrm{fk}}(Q)$ terminates and
\begin{enumerate}
\item every atom it adds has fresh set variables in all non-key positions, so $M$ is unchanged and
      $w(\mathit{chase}_{\mathrm{fk}}(Q))=w(Q)$, $\mathit{kw}(\mathit{chase}_{\mathrm{fk}}(Q))=\mathit{kw}(Q)$.
\item for every database $\mathcal D$ satisfying the keys and foreign keys,
      $Q(\mathcal D)=\mathit{chase}_{\mathrm{fk}}(Q)(\mathcal D)$.
\end{enumerate}
\end{lemma}

\begin{proof}[Proof sketch]
Termination is the standard acyclic-inclusion-dependency argument: each added atom belongs to a target
relation strictly lower in the (acyclic) reference order, so no chain of additions is longer than the
schema's reference depth. The added positions are fresh existential variables, hence set variables, so
no multiset variable and no key occurrence is created, giving~(1). For~(2), on a database satisfying the
foreign keys the witness atom required by each added subgoal already exists, so the added subgoals are
redundant and the satisfying assignments---hence the answer multisets---coincide.
\end{proof}

\fkbound*

\begin{proof}
By \cref{lem:fk-chase}(2), inequivalence over legal databases is inequivalence of $\widehat Q,
\widehat{Q'}$ over the databases satisfying the keys alone, to which \cref{thm:keyed-bound} applies.
The witness it returns
satisfies the keys, and---because the chase made every referenced atom explicit---is closed under the
foreign keys without new multiset structure, keeping the size bound. The key-width and the anchoring
requirement are preserved by \cref{lem:fk-chase}(1), whose added atoms place only fresh set variables in
non-key positions.
\end{proof}
\npcor*

\begin{proof}
Membership. Guess mappings between the terms of $\widehat Q^{\natural}$ and
$\widehat{Q'}^{\natural}$ in both directions and check conditions (1)--(4) of
\cref{sec:relational}, all by inspection of the atoms. By clause~(1) of \cref{thm:keyed-bound},
applied to the chased queries as in \cref{thm:fk-bound}, such mappings exist exactly when
$Q\equiv Q'$. The certificate is polynomial, and so is its subject: the foreign-key chase adds at
most one atom per atom and per declared foreign key, and terminates within the reference depth of
the schema (\cref{lem:fk-chase}); the key-reduct is the set $A$ of statically anchored variables,
a closure under the key dependencies (\cref{prop:anchored}); and the key-chase merges two classes
of terms per step, so at most $|Q|$ times the maximum arity steps occur.

Hardness. Let $Q,Q'$ have $M=M'=\emptyset$ over the schema whose every key is the whole tuple,
where every database is legal and every query is key-anchorable. Exactly one assignment of
$X\cup M$ produces each answer, so $Q(\mathcal D)$ is the set that $Q$ returns under set
semantics, and $Q\equiv Q'$ is equivalence of conjunctive queries, which is
\NP-hard~\cite{ChandraM77}.
\end{proof}

\subsection{Queries with Comparisons}
\label{app:comp-proofs}

\compsound*

\begin{proof}
Suppose $\mu\colon Q'\to Q$ respects regions. As in \cref{sec:relational}, $\mu$ injects the
$Q$-assignments producing $\bar a$ into the $Q'$-assignments producing $\bar a$ through
$\gamma\mapsto\gamma\circ\mu$, and respecting regions makes this preserve the comparisons: a satisfying
$\gamma$ of $Q$ puts $\gamma(\mu y')$ in $R_{\mu(y')}\subseteq R_{y'}$ when $\mu(y')$ is a variable, and on
a constant of $R_{y'}$ otherwise. Hence $\mathrm{mult}_Q\le\mathrm{mult}_{Q'}$ pointwise, and maps in both
directions give equality.
\end{proof}

\comppin*

\begin{proof}
Suppose not, say some slot lies in $R_y\setminus R_{y'}$, and pick a value $v$ in it, off $\mathcal C$. Let
$\mathcal D$ be the canonical database of $R_Q$, freezing the variables of $Q$ to distinct fresh values off
$\mathcal C$, each inside its own region, with $y\mapsto v$. This is possible by density, the regions being
nonempty. Then $\mathcal D$ has at most $|Q|$ atoms, and the freezing satisfies $Q$, so its head $\bar a$ lies
in $Q(\mathcal D)$. By set-equivalence $\bar a\in Q'(\mathcal D)$, witnessed by an assignment $\delta$ of $Q'$
into $\mathcal D$ that meets every var-const atom of $Q'$. As the freezing is injective and the atoms of
$\mathcal D$ are its images of the atoms of $R_Q$, $\delta$ factors as $(\text{freezing})\circ h$ for a
containment map $h\colon Q'\to Q$. Because homomorphisms respect groups, $h(y')$ is a constant or a variable
of $y'$'s group. That group is $\{y,y'\}$, whose sole $Q$-variable is $y$, and the shared position holds no
constant. So $h(y')=y$ and $\delta(y')=v$, forcing $v\in R_{y'}$ against the choice of $v$. Hence
$R_y\subseteq R_{y'}$, and by symmetry $R_y=R_{y'}$.
\end{proof}

\compmaster*

\begin{proof}
\emph{Strictification.} Let $c$ be boundary-free and $\mathcal D$ any ordered database. Choose a
fresh value $c^-$ strictly between $c$ and the next lower constant of $\mathcal C\setminus\{c\}$,
off $\mathrm{adom}(\mathcal D)$, and let $\mathcal D^-$ rename every occurrence of the value $c$
to $c^-$. Queries compare values only to constants, and $c^-$ lies in the same gap of
$\mathcal C\setminus\{c\}$ as $c$, so the renaming preserves every value's order relation to
every constant the queries mention, and preserves $z\le c$, now met by $c^-<c$. As $c$ neither is
a relational constant nor constrains an $(X\cup M)$-variable, it enters an $(X\cup M)$-image only
through the renaming, applied identically to $Q$ and $Q'$. Hence
$\mathrm{mult}_Q(\mathcal D)=\mathrm{mult}_Q(\mathcal D^-)$, and likewise for $Q'$. So
counter-examples may avoid the boundary-free values, and on databases avoiding them each weak
comparison at such a constant holds iff its strict form does. Every database below is built from
fresh values off $\mathcal C$ and avoids these values, so we may assume every position-group of
the pair is directed, retained, or isolated.

\emph{Pinning.} For every isolated group $\{y,y'\}$, $R_y=R_{y'}$ by \cref{comp:pin}, whose
frozen database uses fresh values off $\mathcal C$.

\emph{Resolution.} Resolve every retained group that is not directed, as in
\cref{comp:certified}. The decomposition displayed there gives, for every $(\mathcal D,\bar a)$,
\[
  \mathrm{mult}_Q(\mathcal D,\bar a)-\mathrm{mult}_{Q'}(\mathcal D,\bar a)
  =\sum_{\sigma}\mathrm{mult}_{Q_\sigma}(\mathcal D,\bar a)-\sum_{\sigma'}\mathrm{mult}_{Q'_{\sigma'}}(\mathcal D,\bar a),
\]
a signed sum over a pool of reducts in which every group is directed, pinned to a single open
slot, or isolated with matched regions, and no variable is frozen at a constant. If
$Q\not\equiv Q'$, the sum is not the zero function.

\emph{The placement and the count.} For reducts $P,P'$ of the pool, place $P$ as follows: each
variable of a directed group in its tight slot, an unconstrained variable of an upper group
above all constants and dually for a lower group; each pinned variable in its slot; each
isolated variable in an open slot of the common region, which exists because regions are not
single constants (\cref{comp:framework}). By order-invariance and all-or-none, the coefficient
of $\mathcal P^{(P)}_*$ in $\mathrm{mult}_{P'}(\pcdb{\vec N}{P})$ counts the
multiset-homomorphisms $P'\to P$ whose images meet the comparisons of $P'$ under this placement
(\cref{prop:cmono-sep}). We claim these are exactly the region-respecting ones. On a directed
group, say upper, a variable $z$ sits just below $\mathrm{ub}(z)$, so its block or value lies in
$R_{y'}=(-\infty,\mathrm{ub}(y'))$ exactly when $\mathrm{ub}(z)\le\mathrm{ub}(y')$, that is,
when $R_z\subseteq R_{y'}$. The uniform direction rules out a spurious match, and a constant
image $c$ lies in $R_{y'}$ iff $\{c\}\in R_{y'}$. On a pinned group the placed slot is the
region of the reduct, so placement-respect and region-respect coincide, and the exclusion of
existential variables from pinned groups keeps a pinned slot from matching an image elsewhere.
On an isolated group the map is forced, and the placed block or value lies in the counterpart's
region because the regions match, so the group imposes no constraint.

\emph{The classes.} For reducts $P,P'$ of the pool write $P\preceq P'$ when there is a
region-respecting multiset-homomorphism $P\to P'$. It is reflexive and transitive, and
$P\preceq P'\preceq P$ implies $P\equiv P'$ by \cref{comp:sound}, so $\preceq$ descends to a
partial order on equivalence classes. Collecting reducts by class, write $a_{[P]}\in\mathbb Z$
for the number of $Q$-reducts in $[P]$ minus the number of $Q'$-reducts. The difference above is
$\sum_{[P]}a_{[P]}\,\mathrm{mult}_P$, and as it is not zero, some $a_{[P]}\ne0$. Let $m$ be the
greatest number of multiset variables among classes with $a_{[P]}\ne0$, and pick such a class
$[P]$ of degree $m$ that is $\preceq$-minimal among them. Since equivalent queries have equal
multiplicity functions, the coefficient of $\mathcal P^{(P)}_*$ in
$\mathrm{mult}_{P'}(\pcdb{\vec N}{P})$ is the same for every member $P'$ of a class. By the
count, it is nonzero only for $P'$ with $m_{P'}\ge m$ and $P'\preceq P$, hence $m_{P'}=m$, so by
minimality only $[P]$ itself contributes, $a_{[P]}$ times the positive self-homomorphism count.

\emph{Extraction.} As $[P]$ has maximal degree, $\mathcal P^{(P)}_*$ is a top-degree
multilinear monomial of the combination with a nonzero coefficient. By the Boolean-corner
extraction of \cref{thm:main-bound} the combination is nonzero at some $\vec N\in\{1,2\}^{m}$,
where $\mathrm{mult}_Q\ne\mathrm{mult}_{Q'}$. The reduct $P$ has at most $|Q|$ atoms and width
at most $w$, so its placed family at the corner has at most $2^{w}|Q|$ atoms. Its values are
fresh and off $\mathcal C$, so the witness also separates the original, unstrictified pair.
\end{proof}

\compmonothm*

\begin{proof}
If both homomorphisms exist, $Q\equiv Q'$ by \cref{comp:sound}. Conversely let $Q\equiv Q'$, so
the multiplicity functions agree on the placed family of $Q$ under the tight placement. All
groups being directed, the count in the proof of \cref{comp:master} applies with no resolved and
no isolated groups: the coefficient of $\mathcal P^{(Q)}_*$ in $\mathrm{mult}_{Q}$ is the number
of region-respecting homomorphisms $Q\to Q$, positive by the identity, and in
$\mathrm{mult}_{Q'}$ it is the number of region-respecting homomorphisms $Q'\to Q$. Equality
makes the latter positive, so a homomorphism $Q'\to Q$ exists. The symmetric argument on the
family of $Q'$ gives one from $Q$ to $Q'$.
\end{proof}

\compisolatedthm*

\begin{proof}
Suppose the relational parts are multiset-homomorphic. Every comparison sits on an isolated
variable, whose image under either homomorphism is its counterpart, with equal regions by
\cref{comp:pin}, and every other variable is unconstrained. Both homomorphisms therefore respect
regions, and $Q\equiv Q'$ by \cref{comp:sound}. Conversely let $Q\equiv Q'$, and place each
isolated variable in an open slot of the common region. By all-or-none, every comparison of
either query accepts each placed block or value whole, so for $R\in\{Q,Q'\}$ the coefficient of
$\mathcal P^{(Q)}_*$ in $\mathrm{mult}_{R}(\pcdb{\vec N}{Q})$ counts all multiset-homomorphisms
of the relational parts $R_R\to R_Q$. As in the previous proof, equality of the multiplicity
functions turns the positive self-count into a multiset-homomorphism $R_{Q'}\to R_Q$, and the
symmetric argument on the family of $Q'$ completes the claim.
\end{proof}

\lipthm*

\begin{proof}
By the support gate, $\mathrm{mult}_Q=\beta\,\mathrm{mult}_{Q_0}$ and
$\mathrm{mult}_{Q'}=\beta'\,\mathrm{mult}_{Q'_0}$, with $\beta,\beta'\in\{0,1\}$. Since
$Q\not\equiv Q'$, some legal $(\mathcal D_0,\bar a_0)$ has $\mathrm{mult}_Q\ne\mathrm{mult}_{Q'}$.
Were one zero and the other positive, that answer would lie in one support but not the other, contradicting
set-equivalence. So both are positive, whence $\beta=\beta'=1$,
$\mathrm{mult}_{Q_0}(\bar a_0)\ne\mathrm{mult}_{Q'_0}(\bar a_0)$, and the pinned variables' forced values
satisfy every comparison.

Identify the pinned variables of $Q_0,Q'_0$ exactly as they coincide at $(\mathcal D_0,\bar a_0)$. The
resulting comparison-free bodies still differ, and remain key-anchorable, since adding equalities only enlarges
the closures. By the comparison-free bound over legal databases (\cref{thm:keyed-bound}) they differ on a
legal database $\mathcal D_1$ of at most $2^{\mathit{kw}}|Q|$ atoms. Its canonical answer assigns
distinct variables distinct values, so the pinned values on $\mathcal D_1$ carry exactly the coincidences we
imposed and no others. Some ordering of those values, the one at $(\mathcal D_0,\bar a_0)$, satisfies every
comparison, including the atoms relating two pinned variables. So by order-blindness an injection
$\sigma$ fixing the constants reorders the values of
$\mathcal D_1$ to satisfy every comparison atom, changing neither body's multiplicity nor legality. On
$\sigma\mathcal D_1$ the pinned comparisons hold, so the support gate makes the multiplicities equal to those
of the bodies, which differ. This is the required counter-example.
\end{proof}

\end{document}